%% file: MCmetric.tex
\documentclass[11pt]{siamart190516}

\usepackage{silence}
\usepackage{amssymb}
\usepackage{enumerate}
\usepackage{lmodern}
\usepackage{subcaption} 

\usepackage{tikz}
\usetikzlibrary{arrows.meta,external}
\usepackage{tikz-network}

\usepackage{wrapfig}

\usepackage{booktabs}

\usepackage[utf8]{inputenc}
\usepackage{microtype}
\usepackage{xfrac}
\usepackage{pgfplots} 
\usepackage{pgfgantt}
\usepackage{pdflscape}
\pgfplotsset{compat=newest} 
\pgfplotsset{plot coordinates/math parser=false}
\newlength\fwidth

\usepackage{graphicx}

\newsiamremark{rem}{Remark}
\newsiamremark{claim}{Claim}
\newsiamremark{obs}{Observation}
\newsiamremark{ass}{Assumption}

\hypersetup{unicode=true}

\newcommand{\prob}{\mathbf{P}}
\newcommand{\E}{\mathbf{E}}
\newcommand{\const}{\mathrm{const}}
\newcommand{\Aht}{A^{(\mathrm{hp})}}
\newcommand{\Ahtb}[1][\beta]{A^{{(\mathrm{hp},#1)}}}
\newcommand{\R}{\mathbb{R}}
\newcommand{\ie}{{\it i.e.\/}}
\newcommand{\eg}{{\it e.g.\/}}

\allowdisplaybreaks

\newcommand{\TheTitle}{A metric on directed graphs and Markov chains based on hitting probabilities}
\newcommand{\TheAuthors}{Z. Boyd, {\it et al.\/}}
\newcommand{\TheShortTitle}{Hitting probability metric}
\headers{\TheShortTitle}{\TheAuthors} 

\newcommand{\UNCMath}{Department of Mathematics, Univ.~North Carolina at Chapel Hill}
\newcommand{\naff}{Department of Statistics and Operations Research, Univ.~North Carolina at Chapel Hill}
\newcommand{\paff}{Carolina Center for Interdisciplinary Applied Mathematics, Department of Mathematics and Department of Applied Physical Sciences, Univ.~North Carolina at Chapel Hill}
\newcommand{\jaff}{Courant Institute of Mathematical Sciences, New York University}
\newcommand{\baff}{Department of Mathematics, University of Utah}

\title{\TheTitle
  \thanks{The authors gratefully acknowledge helpful conversations with Erik Thiede, Charles Matthews, Jonathan Mattingly, Brian Van Koten, and Benjamin Z. Webb. We also thank the anonymous reviewers for the article for various helpful suggestions that improved the draft substantially.  Also, the referees for the 2020 SIAM Workshop on Network Science (NS20) gave some helpful suggestions on early versions of the work. ZMB and PJM were supported by the James S. McDonnell Foundation 21st Century Science Initiative---Complex Systems Scholar Award grant \#220020315, with additional support from the Army Research Office (MURI award W911NF-18-1-0244).  JLM was supported in part by NSF CAREER Grant DMS 13-52353 and NSF Grant DMS 19--09035 and thanks Duke University and MSRI for hosting him during part of the completion of this work. BO acknowledges partial support from NSF DMS 16-19755 and 17-52202. JW was supported by the Advanced Scientific Computing Research Program within the DOE Office of Science through award DE-SC0020427. The content is solely the responsibility of the authors and does not necessarily reflect the views of any of the funding agencies.}
}

\author{
  Zachary M.\ Boyd\thanks{\UNCMath\ (\email{zachboyd@email.unc.edu})} \and 
  Nicolas Fraiman\thanks{\naff} \and
  Jeremy L.\ Marzuola\thanks{\UNCMath} \and 
  Peter~J.~Mucha\thanks{\paff} \and 
  Braxton Osting\thanks{\baff} \and
  Jonathan Weare\thanks{\jaff}
}
\date{\today}

\numberwithin{theorem}{section} 

\SetVertexStyle[InnerSep=0pt,MinSize=0.001cm,LineWidth=0pt,LineOpacity=0] 

\begin{document}

\maketitle



\begin{abstract} 
The shortest-path, commute time, and diffusion distances on undirected graphs have been widely employed in applications such as dimensionality reduction, link prediction, and trip planning. Increasingly, there is interest in using asymmetric structure of data derived from Markov chains and directed graphs, but few metrics are specifically adapted to this task. 
We introduce a metric on the state space of any ergodic, finite-state, time-homogeneous Markov chain and, in particular, on any Markov chain derived from a directed graph. 
Our construction is based on hitting probabilities, with nearness in the metric space related to the transfer of random walkers from one node to another at stationarity. 
Notably, our metric is insensitive to shortest and average walk distances, thus giving new information compared to existing metrics.
We use possible degeneracies in the metric to develop an interesting structural theory of directed graphs and explore a related quotienting procedure.
Our metric can be computed in $O(n^3)$ time, where $n$ is the number of states, and in examples we scale up to $n=10,000$ nodes and $\approx 38M$ edges on a desktop computer.
In several examples, we explore the nature of the metric, compare it to alternative methods, and demonstrate its utility for weak recovery of community structure in dense graphs, visualization, structure recovering, dynamics exploration, and multiscale cluster detection.
\end{abstract}

\section{Introduction}

\subsection{Motivation}
Many finite spaces can be endowed with meaningful metrics. For undirected graphs, the geodesic (shortest path), commute time (effective resistance), and diffusion distance~\cite{lafon04diffusion,coifman05geometric,coifman06diffusion} metrics are widely applied~\cite{coifman05geometric,Liben_Nowell_2003,abraham2010highway}. 
The first two can be naively generalized to directed graphs by summing shortest/average walk length in each direction, whereas the third is specifically undirected.
We know of only one graph metric specifically designed for directed graphs, namely the generalized effective resistance distance developed in~\cite{Young_2016b,Young_2016a}.
Overlaying a metric onto a directed structure is a challenge since, by definition, the metric is symmetric.

A related problem is finding metrics on the state space of a finite-state, discrete time Markov chain. In this case, there is also limited prior work, consisting of mean commute time~\cite{rozinas,chebotarev-2020,choi19resistance} and a constant-curvature metric~\cite{vollering2018}.

Metrics fit into the broader category of dissimilarity measures, with the decision whether to impose all metric axioms being application dependent. When a metric is used, this additional structure can enable various algorithmic accelerations, improved guarantees, and useful inductive biases~\cite{elkan2003kmeans,moore00anchors,hamerly10kmeans,boytsov13prune,pitis2020bias}. Furthermore, the metric structure is a key ingredient in proofs of convergence, consistency, and stability. While mostly settled for undirected graphs~\cite{Osting_2017,singer2012vector,singer2017spectral,trillos2018variational,trillos2016consistency}, the development of such theories for directed graphs (digraphs) and Markov chains is an open research problem.  The first positive result for digraphs appeared recently~\cite{Yuan2020}.

In the present work, we introduce and analyze a new metric for digraphs and Markov chains based on the \emph{hitting probability} from one node to another, by which we mean the probability that a random walker starting at one node will reach the other node before returning to its starting node. By correctly combining these probabilities with the invariant distribution of an irreducible Markov chain, a metric can be constructed. This metric differs from other metrics by being insensitive to walk length, thus measuring information that is, in a sense, orthogonal to commute time, as illustrated in examples. In the special case of undirected graphs and with the scale parameter $\beta=1$ (defined below), the hitting probabilities metric is actually the logarithm of effective resistance/commute time (plus a constant), a striking fact proven in~\cite{doyle2000random}, section 1.3.4. For other values of the scale parameter, the hitting probabilities metric is a new addition to the limited catalogue of undirected graph metrics. 
We illustrate the utility of our metric in several examples, both analytical and numerical, related to graph symmetrization, clustering, structure detection, data exploration, and geometry detection.

\subsection{Our contributions}

Let ${(X_t)}_{t\geq 0}$ be a discrete-time Markov chain on the state space $[n] = \{ 1, \ldots , n\}$ with initial distribution $\lambda$ and irreducible transition matrix $P$, \ie, 
\[
  \prob(X_0 = i) = \lambda_i
  \qquad \textrm{and} \qquad  
  \prob(X_{t+1} = j \mid X_t = i) = P_{i,j}\,.
\]
We emphasize that $X$ is not required to be aperiodic.

Let $\phi \in \R_{+}^n$ be the invariant distribution for $P$, \ie, $P^T \phi = \phi$. 
The \emph{hitting time} (starting from a random state distributed like $\lambda$) for a state $i\in [n]$ is the random variable given by 
\[
  \tau_i := \inf\{ t \geq 1 \colon X_t = i \}\,. 
\]
For $i,j \in [n]$, let us define 
\begin{equation}
  Q_{i,j} := \prob_i [\tau_j < \tau_i]\,,
\end{equation}
which denotes the probability that starting from site $i$ (\ie, the subscript on $\prob_i$ is used to indicate that $\lambda = \delta_{i}$) the hitting time of $j$ is less than the time it takes to return to $i$.  We emphasize that we consider $\tau_j < \tau_i$ here for a single walk and take the probability of such an event over all walks starting at $i$ when computing $Q_{i,j}$.  An expression for the \emph{hitting probability matrix}, $Q$, in terms of the transition matrix will be given in~\cref{eq:Q}; see \cref{s:CompMeth} on computational methods.

\begin{lemma}\label{t:KeyIdentity} The following relationship holds\footnote{\Cref{t:KeyIdentity} was previously (and independently) proven in~\cite{chien04link}, in the context of Markov chain perturbation theory applied to the internet. It was possibly known even earlier.} for $i\ne j$:
  \begin{equation}%
    \label{Qpi}
    Q_{i,j} \phi_i = Q_{j,i} \phi_j\,. 
  \end{equation}
\end{lemma}
The weighting by the invariant measure is motivated by connections between the invariant measure and random walks as found in \cite[Section 1.7]{Norris_1997}.  A proof of \cref{t:KeyIdentity} is given in \cref{s:Proofs}.
\begin{rem}
\Cref{t:KeyIdentity}~implies that, with appropriate choice of $Q_{ii}$, $\frac1n Q$ is a reversible Markov chain with invariant distribution $\phi$. 
\end{rem}
We define the \emph{normalized hitting probabilities matrix},  $\Ahtb  \in \mathbb R^{n \times n}$, by
\begin{equation}
  \label{Aht}
  \Ahtb_{i,j} := 
  \begin{cases}
    \dfrac{ \phi_i^{\beta} }{ \phi_j^{1-\beta} } Q_{i,j} & i \ne j \\
    1   & i=j
  \end{cases}
\end{equation}
where $\beta \in [\sfrac12, \infty)$. In contexts where the choice of $\beta$ is not important, we simply write $\Aht = \Ahtb$. Two useful choices for $\beta$ are $1$ and $1/2$. 
The $Q_{i,j}$ matrix has recently been shown to play a key role in determining the error of a family of stratified Markov chain Monte Carlo methods~\cite{dinner2017stratification,Thiede_2015}. 

From \cref{t:KeyIdentity}, we immediately have the following Corollary. 
\begin{corollary}\label{l:symmetric}
  The matrix $\Ahtb$ defined in~\cref{Aht} is symmetric.    In particular, 
  \begin{equation}
    \Ahtb[\sfrac12]_{i,j} = \sqrt{Q_{i,j}Q_{j,i}}.
    \label{d2sym}
  \end{equation}
\end{corollary}

\begin{proof}
  We observe
  \begin{align*}
    \Ahtb_{i,j} & = \frac{ \phi_i^\beta }{ \phi_j^{1-\beta} } Q_{i,j} =  \frac{ \phi_i^{\beta-1} }{ \phi_j^{1-\beta} }  \phi_i Q_{i,j} \\
    & =  \frac{ \phi_i^{\beta-1} }{ \phi_j^{1-\beta} }  \phi_j Q_{j,i} = \frac{ \phi_j^{\beta} }{ \phi_i^{1-\beta} } Q_{j,i} = A^{(\mathrm{hp},\beta)}_{j,i}.
  \end{align*}
  Hence, $\Ahtb$ is symmetric. 
  
  To prove \eqref{d2sym}, we observe that ${\left(\Ahtb[\sfrac12]_{i,j}\right)}^2 = \frac{\phi_i}{\phi_j} Q_{i,j}^2 = Q_{i,j} Q_{j,i}$ by~\cref{Qpi}.
\end{proof}

In some applications, information about relatedness of vertices in a graph will be most immediately encoded in the form of a non-stochastic adjacency matrix $A$.  In this case the input adjacency matrix can be transformed into a stochastic matrix $P$ either by a similarity transformation involving the dominant right eigenvector of $A$ or by normalization of the rows of $A$ so that they sum to 1.  The resulting stochastic matrix $P$ can then be used as in~\cref{Aht} to construct $A^{(\textrm{hp},\beta)}$, itself a symmetric adjacency matrix on the vertices of the network.  In this article we do not address the relative merits of methods to transform an adjacency matrix into a stochastic matrix. We use row normalization unless otherwise stated.

Given an irreducible stochastic matrix $P$, we can thus define a distance $d^\beta \colon [n] \times [n] \to \mathbb R$, which we refer to as the \emph{hitting probability metric}, by 
\begin{equation} \label{e:Dist}
  d^\beta(i,j) = - \log \left( \Ahtb_{i,j} \right). 
\end{equation}

\begin{theorem}\label{t:Metric}
  The hitting probability metric, $d^\beta \colon [n] \times [n] \to \mathbb R$, defined in~\cref{e:Dist} is a metric for $\beta \in (\sfrac12, 1]$.  For $\beta = \sfrac12$, $d^\beta$ is a pseudo-metric\footnotemark and there exists a quotient graph
  on which the distance function becomes a metric. 
\end{theorem}
\footnotetext{Recall that a pseudo-metric on $[n]$ is a non-negative real function $f\colon [n] \times [n] \to \mathbb R_{\geq 0}$ satisfying $d(i,i) = 0$, symmetry $d(i,j) = d(j,i)$, and the triangle inequality $d(i,j) \leq d(i,k) + d(k,j)$. A pseudo-metric is a metric if we can identify indiscernible values, \ie, $d(i,j) = 0 \iff i=j$.}
In~\cref{t:12bounds}, we show that there exists a quotient graph on which $d^{\sfrac12}$ is a metric and which preserves many of the metric properties of the original graph.\footnote{While the usual pseudo-metric quotienting procedure could apply here, there is no guarantee that there would be a corresponding subgraph, which is why~\cref{t:12bounds} is needed.}  The key observation for the $d^{\sfrac12}$ pseudometric is that in order for two vertices to be distance $0$ from each other, the probability of hitting the other vertex before returning must be $1$ for both.  Hence we provide (in~\cref{structure,quotients}) a means of effectively collapsing these vertices to a single vertex, carefully preserving the overall probabilities relative to the remaining vertices.

\begin{rem}
  In light of~\cref{t:KeyIdentity,t:Metric}, $\Aht$ has two interpretations, first as a symmetrization of $A$, and second as a weighted similarity graph corresponding to $d$, since $A(i,j) = e^{-d(i,j)}$. The practice of associating a finite (subset of a) metric space with a similarity graph in this way is widespread, especially in the manifold learning and graph-based machine-learning communities.\footnote{\cite{self-tuning} cites~\cite{shi-malik,relocalization} as this similarity function's first use specifically for graph-based clustering.} Thus, in our experiments, we favored the use of $\Aht$ for certain applications where it seemed more natural.
\end{rem}

Finally, we show how advances from~\cite{Thiede_2015} enable us to compute the distance matrix in $O(n^3)$ operations, allowing us to scale up to $\approx 38M$ edges in examples on a 
Lenovo ThinkStation P410 desktop with Xeon E5--1620V4 3.5 GHz CPU and 16 GB RAM using {\sc MATLAB} R2019a Update 4 (9.6.0.1150989) 64-bit (glnxa64).
We also provide various synthetic examples to help develop an intuition for the metric and its differences from other measures. We conclude with an example using New York City taxi data to illustrate how our metric can aid in data exploration.

\subsection{Relationship to other notions of similarity and metrics}\label{s:RelWork}
In this section, we discuss some related notions of similarity and metrics on finite state spaces with asymmetric (directional) relationships. Our focus is on symmetric notions of dissimilarity, with an emphasis on metrics. While, in some applications, asymmetric similarity scores may be the right choice (see, \eg, Tversky's seminal  work on features of similarity~\cite{Tversky_1977}), we restrict our scope to symmetric notions. We do, however, wish to mention directed metrics (also called quasi-metrics), which are a natural analogue to metric spaces for relaxations of digraph cut problems~\cite{directed_metric}.

From~\cite{chebotarev-2020,kemeny,rozinas}, we know that commute time is a metric on ergodic Markov chains. In~\cite{choi19resistance,Young_2016b,Young_2016a}, generalizations of effective resistance are developed for ergodic Markov chains and directed graphs. Commute time and resistance-based metrics are popular and more robust than shortest-path distances, although they are not informative in certain large-graph limits~\cite{vonluxborg2014}. In~\cref{sec:gluedcycles}, we compare the effective resistance of~\cite{Young_2016b,Young_2016a} to the hitting probability metric on a particular example.

In~\cite{vollering2018}, a metric is developed on Markov chains. This metric gives the chain constant curvature, in an appropriate generalized sense. Distance in this metric is then related to the distinguishability after one step of random walks beginning at the two distinct nodes. The metric is constructed jointly with the curvature using a fixed point argument. It is expected to be useful in proving, for example, concentration inequalities for Markov chains.

Notions of diffusion distance to a set $B$ on undirected graphs have been explored recently for the connection Laplacian~\cite{singer2012vector} and for the graph Laplacian~\cite{cheng2019diffusion}.  The notion of distance is determined by taking $\ell$ steps using the random walk generated by the symmetric graph adjacency matrix $A$ with degree matrix $D$, i.e.\ it counts the number of walks of  length $2\ell$ from $i$ to $j$.
Diffusion distances from a vertex $i$ to a sub-graph $B$ in~\cite{cheng2019diffusion} is defined as the smallest number of steps for all random walks started at $i$ to reach $B$.  The work~\cite{singer2012vector} established that diffusion distances converge to geodesic distances in the high density limit of random graphs on manifolds, and \cite{cheng2019diffusion} explored how eigenvectors relate to this notion of distance.  Directed graphs have been represented as magnetic connection Laplacians on undirected graphs through a notion of polarization, see~\cite{fanuel2018magnetic}, after which a version of diffusion distance can be applied.

A variety of methods exist in machine learning to compute ``graph representations,'' which are learned embeddings of nodes, subgraphs, or entire (possibly directed) graphs into Euclidean space so that they can be fed into standard machine learning tools~\cite{hamilton17representation}. These can be seen as imposing a metric on directed graphs, with the main drawbacks relative to the hitting probability metric being model complexity, difficulty of interpretation, and difficulty of analysis.

In~\cite{malliaros13clustering}, existing symmetrization techniques for directed graphs are surveyed. In particular, we mention \cite{satuluri2011symmetrizations,zhou2005semi,lai2010extracting,chen2008clustering}.  In each of these articles, clustering, community detection and/or semi-supervised learning techniques are considered on directed graphs using various symmetrizations, such as that of Fan Chung (\eg~\cite{satuluri2011symmetrizations}) or using commute times similar to those in the effective resistance metric (\eg~\cite{chen2008clustering}). Our results use $\Aht$ as a symmetrization, and we will see that this enables us to perform the tasks just mentioned, although with different and sometimes more helpful results.

In~\cite{fitch}, the metric of~\cite{Young_2016a,Young_2016b} is used as the basis for a digraph symmetrization technique. It is guaranteed to preserve effective resistances, possibly relying on negative entries. Rigorous applications to directed cut and graph sparsification are given.

\subsection*{Outline} 
We prove \cref{t:KeyIdentity,l:symmetric,t:Metric} in \cref{s:Proofs}.
In \cref{s:CompMeth}, we describe computational methods to compute the normalized hitting probabilities matrix, $\Ahtb$. 
In \cref{s:Examples}, we give some examples of the computed metric. 
We conclude in~\cref{s:Disc}.

\section{Proofs and discussion of structural properties}\label{s:Proofs} 
\subsection{Structure of the normalized hitting probabilities matrix}
\begin{proof}[Proof of \cref{t:KeyIdentity}]
  The probability that $X_t$ starts from $i$  and hits $j$ at least $k+1$ times before returning to $i$ can be expressed as
  \[
    \prob_i [\tau_j < \tau_i] \prob_j {[\tau_j < \tau_i]}^k\,. 
  \]  
  We let $V_i^j$ be the number of times $X_t$ hits $j$ before returning to $i$, 
  $
    V_i^j = \sum_{t=1}^{\tau_i} {1}_{X_t = j}\,.
  $
  Then, we have
  \[
    \prob_i [\tau_j < \tau_i] \prob_j {[\tau_j < \tau_i]}^k = \prob_i [V_i^j \geq k+1]\,.
    \label{r1}
  \]

  Now observe that 
  \begin{align}
    &\label{geo} \sum_{k=0}^\infty \prob_i [\tau_j < \tau_i] \prob_j {[\tau_j < \tau_i]}^k  =\sum_{k=0}^\infty \prob_i [ V_i^j \ge k + 1] 
    =\sum_{k=0}^\infty \E_i \left[1_{V_i^j \ge k+1}\right]  \\
    &\hspace{1cm} = \E_i \left[\sum_{k=0}^\infty  {1}_{V_i^j \geq k+1}\right]
    = \E_i \left[\sum_{k=0}^\infty k 1_{V_i^j=k} \right]  
    = \sum_{k=0}^\infty k \prob[V_i^j=k] \notag \\
    &\hspace{1cm}= \E_i [ V_i^j ] \label{r3}\,.
  \end{align}
  The expectation in \cref{r3} is known to satisfy
  \begin{equation}
    \label{gammaid}
    \E_i [ V_i^j ] = \frac{\phi_j}{\phi_i}\,,
  \end{equation}
  which is proved in, for example,~\cite[Theorem 1.7.6]{Norris_1997}.

  However, we recognize the expression \cref{geo} as a geometric series and hence have
  \begin{align*}
    \sum_{k=0}^\infty \prob_i [\tau_j < \tau_i] \prob_j {[\tau_j < \tau_i]}^k
    &= \prob_i [\tau_j<\tau_i]{\left(1-\prob_j[\tau_j<\tau_i]\right)}^{-1} \\
    &= \prob_i [\tau_j < \tau_i] \prob_j {[\tau_i < \tau_j]}^{-1} = Q_{i,j} Q_{j,i}^{-1}\,.
  \end{align*}
  
  Combining this with~\cref{gammaid} we arrive at $Q_{i,j} Q_{j,i}^{-1}=\frac{\phi_j}{\phi_i}$.
\end{proof}

To prove~\cref{t:Metric}, we will need one more lemma
\begin{lemma}%
  \label{Qlemma}
  The following inequality holds
  \begin{equation}
    \label{Qineq}
    Q_{i,j} \geq Q_{i,k} Q_{k,j}\,.
  \end{equation}
\end{lemma}


\begin{proof}
  Consider the corresponding auxiliary Markov process restricted to nodes $i$, $j$, and $k$ with $3\times 3$ transition matrix $F$, the elements of which we denote by, \eg, $F_{i,j} = \prob_i[\tau_j < \min \{ \tau_i, \tau_k\}]$ and $F_{i,i} = \prob_i[\tau_i < \min \{\tau_j, \tau_k\}]$. That is,  $F_{i,j}$ gives the probability of a random walker starting at $i$ eventually reaching $j$ before either reaching $k$ or returning to $i$, while $F_{i,i}$ gives the probability of a random walker starting at $i$ returning to $i$ before reaching either $j$ or $k$.
  Since $F_{i,i}+F_{i,j}+F_{i,k}=1$, we have
  \begin{equation*}
    Q_{i,j} = F_{i,j} + \frac{ F_{i,k}F_{k,j}}{1-F_{k,k}}\,, \quad
    Q_{i,k} = F_{i,k} + \frac{ F_{i,j}F_{j,k}}{1-F_{j,j}}\,, \quad
    Q_{k,j} = F_{k,j} + \frac{ F_{k,i}F_{i,j}}{1-F_{i,i}}\,.
  \end{equation*}
  Hence, we observe
  \begin{align*}
    Q_{i,k} Q_{k,j} & = \left( F_{i,k} + \frac{F_{i,j} F_{j,k}}{1-F_{jj} }  \right) \left( F_{k,j} + \frac{F_{k,i} F_{i,j}}{1-F_{ii} }  \right) \\
    & =   F_{i,k}F_{k,j} + F_{i,j}  \left(  \frac{F_{j,k} F_{k,j}}{1-F_{jj} }  + \frac{F_{i,k} F_{k,i}}{1-F_{ii} } + \frac{F_{j,k} F_{k,i} F_{i,j}}{(1-F_{ii}) (1-F_{jj} ) }  \right)\,.
  \end{align*}
  Using
  \[
    \frac{F_{j,k}}{1-F_{j,j}} = \frac{F_{j,k}}{F_{j,i} + F_{j,k}} < 1\,,
  \]
  we then observe
  \begin{align*}
    Q_{i,k} Q_{k,j} & \leq F_{i,k}F_{k,j} + F_{i,j}  \left( F_{k,j}  + \frac{F_{i,k} F_{k,i}}{1-F_{ii} } + \frac{ F_{k,i} F_{i,j}}{1-F_{ii}}  \right)  \\
    &= F_{i,k}F_{k,j} + F_{i,j}  \left( F_{k,j}  + F_{k,i} \frac{F_{i,k} + F_{i,j} }{1-F_{ii} }  \right)  \\
    & = F_{i,k}F_{k,j} + F_{i,j}  \left( F_{k,j}  + F_{k,i}   \right)  \\
    & = \left( \frac{F_{i,k} F_{k,j} }{1-F_{k,k}} + F_{i,j}  \right) (1-F_{k,k}) \\
    & = Q_{i,j}(1-F_{k,k})  \leq Q_{i,j}\,.
  \end{align*}

\end{proof}

\subsection{Hitting probability metric}
In this section we establish~\cref{t:Metric}.  In particular, we explore the notion that, much like effective resistance, the normalized hitting probabilities matrix provides a natural notion of distance on the digraph (or between states of a Markov Chain).  

To begin, we recall the definition of $d^\beta$ from~\cref{e:Dist} and note that we have already established the symmetry $d^\beta(i,j) = d^\beta (j,i)$ for all $i,j$.  
As seen from the statement of~\cref{t:Metric}, we will observe that the triangle inequality holds for all $\beta \ge \sfrac12$ and that positivity holds for all $\beta > \sfrac12$.  In the case $\beta = \sfrac12$, $d^{\sfrac12}$ gives a pseudometric structure, as there can indeed exist structures in a directed graph or Markov Chain on which $d^{\sfrac12} (i,j) = 0$ and $i \neq j$.  As an example, consider the nodes on a cycle with in degree and out degree $1$. (See~\cref{s:Examples}.)

When $d^{\sfrac12}(i,j)=0$ there are specific structures that restrict all random walks leaving $i$ so that they must hit $j$ before returning to $i$.  We show in~\cref{t:12bounds} that for any graph, there exists a canonical quotient graph on which $d^{\sfrac12}$ is indeed a metric that is closely related to $d^{\sfrac12}$ on the original graph.  Let us now proceed to the proofs.

\begin{proof}[Proof of \cref{t:Metric}]

  First, we show positivity for $i\ne j$. Note that $d^\beta(i,j) > 0$ iff $\Ahtb_{i,j} < 1$. Consider the converse 
  \[
    1 \le \Ahtb_{i,j} = \frac{\phi_i^\beta}{\phi_j^{1-\beta} }  Q_{i,j} = \frac{\phi_j^\beta}{\phi_i^{1-\beta} }  Q_{j,i}\,,
  \]
  that is,
  \[
    \frac{\phi_j^{1-\beta}}{\phi_i^\beta} \le Q_{i,j} \quad \text{ and } \quad \frac{\phi_i^{1-\beta}}{\phi_j^\beta} \le Q_{j,i}\,.
  \]
  Then if $\beta > \sfrac12$, we have
  \[
    1 \ge Q_{i,j} Q_{j,i} \ge \phi_j^{1-2\beta} \phi_i^{1-2\beta} > 1\,,
  \]
  a contradiction. For $\beta=\sfrac12$, the last inequality above becomes an equality, so the corresponding argument by contradiction requires only that $\Ahtb[\sfrac12] \le 1$ and thus $d^{\sfrac12}(i,j) \ge 0$.

  Symmetry follows from~\cref{l:symmetric}, and $d^\beta(i,i)=0$ is immediate, so all that remains is the triangle inequality.

  To prove the triangle inequality, we observe for $i\neq j\neq k$ that
  \begin{align}
    d^\beta(i,j) &= - \log \left( \Ahtb_{i,j} \right) = -\beta\log\phi_i -(\beta-1)\log\phi_j - \log Q_{i,j} \notag \\
    &= d^\beta(i,k) + d^\beta(k,j) + (2 \beta - 1) \log\phi_k + [\log Q_{i,k} + \log Q_{k,j} - \log Q_{i,j}]\,,
    \label{e:triangle-slack}
  \end{align}
  which, applying~\cref{Qlemma}, proves that the triangle inequality holds for all $\beta \geq \sfrac12$.
\end{proof}

We observe that the $2 \beta -1$ coefficient of $\log \phi_k$ in~\cref{e:triangle-slack} vanishes when $\beta=\sfrac12$, hence the triangle inequality is as tight as possible (since $Q_{i,k}Q_{k,j}/Q_{i,j} = 1$, \eg,~for a directed cycle graph).  
From the above argument, we can see that the only obstruction to $d^{\frac12}$ being a metric is if there is a pair $i,j$ such that
\[
  \Ahtb[\sfrac12]=\sqrt{ \frac{\phi_j}{\phi_i} } Q_{j,i}  = \sqrt{ \frac{\phi_i}{\phi_j} } Q_{i,j} = 1\,.
\]
In this case,
\[  Q_{j,i} = \sqrt{ \frac{\phi_i}{\phi_j} }  \quad\text{ and }\quad Q_{i,j} = \sqrt{ \frac{\phi_j}{\phi_i} }\,. \]
Thus, $Q_{i,j} Q_{j,i} =1$,
which, as they are both probabilities, means in fact $Q_{i,j} = Q_{j,i} =1$.
Hence, also $\phi_i = \phi_j$.

\begin{obs}
  The condition $\phi_i = \phi_j$ is not an extra restriction beyond $Q_{i,j}  = Q_{j,i} =1$: if $Q_{i,j} =Q_{j,i}=1$ then a random walker must visit $i$ every time it visits $j$ (and vice versa) and hence the invariant probabilities of sites $i$ and $j$ must be equivalent.
\end{obs}

\subsection{Structure theory of digraphs where \texorpdfstring{$d^{\sfrac12}$}{d} is not a metric}%
\label{structure}
In this subsection, we investigate the structure of graphs where $d^{\sfrac12}$ is not a metric, which we refer to as ($d^{\sfrac12}$-) degenerate. This is useful for understanding our metric embedding and is foundational for~\cref{quotients}, where we derive the quotienting procedure to repair graph degeneracies. In this section, we first give a general construction to produce degenerate graphs and show that all degenerate graphs can be constructed in this way. Next, we give a general decomposition of degenerate graphs into equivalence classes and their segments.

\subsubsection{A general construction for degenerate graphs}

A simple example of a graph with $Q_{i,j}=Q_{j,i}=1$ is a closed, directed cycle.  However, we also have the following much more general construction: 
Take any two directed acyclic graphs (DAGs), $G_1$ and $G_2$. Connect all the leaves (sinks/nodes of our-degree zero) of $G_1$ to $i$ and all the leaves of $G_2$ to $j$. Connect $j$ to all the roots (sources/nodes of in-degree 0) of $G_1$ and $i$ to all the roots of $G_2$. Possibly add edges between $i$ and $j$. Then, for each node $k$ except $i$ and $j$, replace it with an arbitrary strongly connected graph $H_k$ (corresponding to an irreducible Markov chain), replacing each edge to (from) $k$ with at least one edge to (from) a node in $H_k$. The resulting graph is strongly connected and has $i$ and $j$ only reachable through each other.

In fact, all graphs with $Q_{i,j}=Q_{j,i}=1$ can be constructed this way. To see this, note that $i$ and $j$ must have positive in and out degree by strong connectedness. Let $C_i$ be those nodes reachable from $i$ without passing through $j$, and define $C_j$ similarly. Consider the following claim:
\begin{claim} 
  $C_i \cup C_j \cup \{i,j\}$ includes all nodes of the graph, and $C_i\cap C_j$ is empty. 
\end{claim}
\begin{proof}
  For the first part, consider a fixed node $k$ which is not $i$ or $j$, together with a shortest path $C_{ik}$ from $i$ to $k$. If $C_{ik}$ does not pass through $j$, then $k\in C_i$; otherwise, $k$ is reachable from $j$ without passing through $i$ since $C_{ik}$ is a shortest path.

  For the second part, assume otherwise; that is, pick $k \in C_i \cap C_j$. Then there exist (1) a path $C_{ik}$ from $i$ to $k$ not passing through $j$, (2) a path $C_{jk}$ from $j$ to $k$ not passing through $i$, and (3) a path $C_{kj}$ from $k$ to $j$ (by strong connectedness). Assume WLOG that $C_{kj}$ passes through $j$ only at the end. $C_{kj}$ cannot pass through $i$ since otherwise $C_{ik}+C_{kj}$ contains a walk from $i$ to $i$ without passing through $j$, violating $Q_{i,j}=1$.
  But then $C_{jk}+C_{kj}$ would be a walk from $j$ to $j$ that does not pass through $i$, contradicting $Q_{j,i}=1$.
\end{proof}
Since $C_i$ and $C_j$ can only be connected through $i$ and $j$, removing $i$ and $j$ disconnects these two sets. Now consider the subgraphs induced by $C_i$ and $C_j$, respectively. As can be done with any directed graph, we reduce each of these subgraphs to their quotients under the mutual reachability equivalence relation, yielding a pair of DAGs. The next subsection generalizes this decomposition to account for all nodes for which $d^{\sfrac12}$ vanishes rather than a single pair.

\subsubsection{Decomposition into equivalence classes and segments}

Consider an equivalence class $\alpha=\{a_1,a_2,\ldots,a_K\}$ of nodes under the equivalence relation $i\sim j \Leftrightarrow d^{\sfrac12}_{i,j}=0$.
We refer to a node in a non-singleton equivalence class as \emph{($d^{\sfrac12}$-) degenerate}. A graph is $d^{\sfrac12}$-degenerate if it has a degenerate node.
\begin{definition}\leavevmode
  \begin{itemize}
    \item A \emph{walk} is a sequence of nodes $\{i_1,i_2,\ldots,i_K\}$ such that $P_{i_k,i_{k+1}}\!\!>\!0$, for $1 \le k<K$.
    \item A walk is \emph{closed} if $i_K = i_1$.
    \item A closed walk is a \emph{commute} from $i_1$ if $i_{k} \ne i_1$, for $1<k<K $.
    \item A walk is a \emph{path} if $i_k\ne i_{k'}$ when $k'\ne k$. A commute is a \emph{cycle} if it is a path when the last element is removed.
  \end{itemize}
\end{definition}
\begin{lemma}
  A commute from $a_k\in \alpha$ must include each of the other members of $\alpha$ exactly once in an order that depends only on the graph.
\end{lemma}
\begin{proof}
  The proof is in three assertions. We assume without loss of generality that commutes are from $a_1$.
  First, each $a_k$ is visited. This is the same as claiming that $Q_{a_1,a_k} = 1$, which was shown in the proof of~\cref{t:Metric}. 
   Second, each $a_k$ is visited at most once, since $Q_{a_k,a_1} = 1$.
   Lastly, each $a_k$ is visited in a fixed order: Let $J_1$ and $J_2$ be commutes from $a_1$ that visit, respectively, $a_2$ before $a_3$ and vice versa. Then let $J_1'$ and $J_2'$ be the sub-walks from $a_1$ to $a_2$ and from $a_2$ to $a_1$ in $J_1$ and $J_2$, respectively. The concatenation of $J_1'$ and $J_2'$ is thus a commute from $a_1$ that does not visit $a_3$, a contradiction.
\end{proof}
In the rest of~\cref{structure}, we assume that equivalence classes under $\sim$ are sorted so that they must be visited in the order by commutes from their first element. Similarly, if the $K$ elements of an equivalence class are numbered $a_1,\ldots,a_K$, we naturally identify $a_x = a_{x \!\!\!\mod\! K}$.
\begin{lemma}%
  \label{seg_lem}
  Given an equivalence class $\alpha$ under $\sim$, for each $j\in G - \alpha$ there is a unique $k$ such that all walks $J$ containing $j$ with $\alpha\cap J \ne \varnothing$ include either $a_{k-1}$ before $j$ or $a_{k}$ after $j$.
\end{lemma}
\begin{proof}
  It is enough to consider paths. Suppose $J_1$ and $J_2$ are two paths from $j$ which reach $a_k$ and $a_{k'}$, respectively, before reaching any other elements of $\alpha$, with $k\ne k'$. By strong connectivity, we can select a shortest path from $a_k$ to $j$ to extend $J_1$ to a cycle from $j$, which we call $J_3$. That is, if we select a shortest (in number of distinct steps) path, $\gamma$, from $a_k$ to $j$ then $J_1 \cup \gamma= J_3$ is the required extension of $J_1$.

  Now, $J_3$ can be cyclically reordered to be a commute from $a_k$. Thus, $J_3$ includes $a_{k'}$, and since $\gamma$ was shortest possible, it includes $a_{k'}$ exactly once. Let $J_4 \subset J_3$ be the sub-walk from $a_{k'}$ to $j$. Then concatenating $J_4$ and $J_2$ gives a commute from $a_{k'}$ that does not include $a_k$, a contradiction. Thus, $j$ has a unique successor $a_k$ in $\alpha$.

  The conclusion that there is a unique predecessor of $j$ in $\alpha$ follows by reversing the direction of all edges and re-applying the above argument. It must be $a_{k-1}$ since $a_k$ is the first member of the equivalence class encountered in any commute from $j$.  \end{proof}

\begin{definition}
  We will here refer to the equivalence classes on $G-\alpha$ induced by \cref{seg_lem} as \emph{($\alpha$-) segments} of $G$.\footnote{Alternatively, we could define segments more generally with respect to any node set $\alpha$. Then the segment corresponding to $i \in \alpha$ is the set of nodes reachable from $i$ without passing through any other elements of $\alpha$. From this perspective, the absolute segments described later are simply the intersection of the segments with respect to all the equivalence classes.}
\end{definition}
\begin{lemma}
  Given distinct equivalence classes $\alpha$ and $\beta$ under $\sim$, every element of $\alpha$ must lie within a single segment induced by $\beta$.
\end{lemma}
\begin{proof}
  Let $\alpha = \{a_1,a_2,\ldots,a_{K_{\alpha}}\}$ and $\beta = \{b_1,b_2,\ldots,b_{K_{\beta}}\}$.
  Suppose, by way of contradiction, that $a_1$ lies between $b_{k_1}$ and $b_{k_1+1}$ and $a_2$ lies between $b_{k_2}$ and $b_{k_2+1}$ for $k_1 \ne k_2$. By strong connectedness, there exists a (shortest) path from $b_{k_1}$ to $a_1$ to $b_{k_1+1}$. If $b_{k_1+1}$ and $b_{k_2}$ are distinct nodes, there also exists a shortest path from $b_{k_1+1}$ to $b_{k_2}$. Since $Q_{b_{k_2},a_2} < 1$, there exists a shortest path from $b_{k_2}$ to $b_{k_2+1}$ not passing though $a_2$. Finally, if $b_{k_2+1}$ and $b_{k_1}$ are distinct nodes, there exists a shortest path from $b_{k_2+1}$ to $b_{k_1}$. Concatenating all these paths gives a commute from $a_1$ to itself not passing though $a_2$, a contradiction.
\end{proof}
The foregoing lemmata show that the nontrivial equivalence classes in a $d^{\sfrac12}$-degenerate digraph induce a structure of equivalence cycles and their segments, with distinct equivalence cycles restricted to lie within the segments of each other. This has potential application in segmentation of directed graphs and will be an important technical tool in the proofs in the next subsection.

\subsection{Quotients of \texorpdfstring{$d^{\sfrac12}$}{d}-degenerate Markov chains}%
\label{quotients}
Next, we develop a way to transform a Markov chain $X$ for which $d^{\sfrac12}$ is not a metric into a quotient Markov chain $X'$, for which $d^{\sfrac12}$ is a metric. 
\begin{rem}
  In~\cref{quotients}, we identify singleton classes with their member.
  Additionally, we append a prime to any symbol when it is meant to refer to $X'$ rather than $X$.
\end{rem}
The quotient graph is given by the following construction, which has appeared in~\cite{mitavskiy08} as well as in~\cite{madras_2002_decomposition,martin_2000_staircase,caracciolo_1992_tempering}, and possibly other places. 
\begin{definition}%
  \label{d:quotient}
  Given a Markov chain $X$ and an equivalence relation on the states of $X$, the \emph{quotient Markov chain} has one state for each equivalence class, and the transition probabilities are given by 
  \[
    P_{U,V}' = \frac{1}{\phi_U} \sum_{i\in U} \phi_i P_{i,V}
    =  \frac{1}{\phi_U} \sum_{i\in U} \sum_{j\in V} \phi_i P_{i,j}\,,
  \]
  where $\phi_U = \sum_{i \in U} \phi_i$.
\end{definition}
The map that sends $X$ to $X'$ is denoted $\iota$.
It can be shown~\cite{mitavskiy08} that the invariant measure on $X'$ evaluated at state $U$ is $\phi'_U = \phi_U$. 
Furthermore, $P$ carries information about the equilibration rate in ergodic chains~\cite{martin_2000_staircase,madras_2002_decomposition,caracciolo_1992_tempering}, although we do not use this fact in this paper.
When applying~\cref{d:quotient} to $\sim$, the definition reduces to
\[
  P_{U,V}' =  \frac{1}{|U|} \sum_{i\in U} \sum_{j\in V} P_{i,j}\,,
\]
since $\phi$ is constant within equivalence classes (see proof of~\cref{t:Metric}).
\begin{lemma}[Quotienting one class at a time]
  Let $\sim$ induce the non-singleton classes $\{\alpha_1,\alpha_2,\ldots,\alpha_L\}$. For a node set $S$, let $\sim_{S}$ be the relation with non-singleton class $S$, keeping all other nodes in individual (singleton) classes. One can then produce a graph with the same nodes as $P'$ by performing a series of quotienting operations $P \underset{\sim_{\alpha_1}}{\to} P_1 \underset{\sim_{\alpha_2}}{\to} \cdots \underset{\sim_{\alpha_L}}{\to} P_L$. Then $P_L = P'$, after identifying nested classes with the nodes in them, \eg,  $\{\{a\},\{b\}\} \to \{a,b\}$.%
  \label{l:order}
\end{lemma}
\begin{proof}
  The proof is by induction on $L$. If $L=1$, the result is vacuously true. So assume the result is true for graphs having $L$ non-singleton equivalence classes, and we proceed to establish the result for $L+1$ non-singleton classes. Let $G$ have classes $\{\alpha_1,\ldots,\alpha_{L+1}\}$ under $\sim$. Then we apply $\sim_{\alpha_1}$ to get $P_1$ and then use the inductive assumption to conclude that $P_{L+1} = P_1'$. So we need to prove that $P_1'=P'$. We have
  \[
    {P_1'}_{\alpha,\beta} = 
    \begin{cases}
      P_{\alpha,\beta}                                    &   \alpha\ne\alpha_1,\beta\ne\alpha_1 \\
      \sum_{j\in\beta} {P_1}_{\alpha_1,j}                 &   \alpha=\alpha_1,\beta\ne\alpha_1    \\
      \frac1{|\alpha|} \sum_{i\in\alpha} {P_1}_{i,\alpha_1} &   \alpha\ne\alpha_1,\beta=\alpha_1\\
      {P_1}_{\alpha_1,\alpha_1}                             &   \alpha=\alpha_1=\beta\, ,
    \end{cases}
  \]
  where we have implicitly used the fact that $\phi_{P_1}$ has the form given in~\cref{l:order}.
  Expanding further gives
  \[
    {P_1'}_{\alpha,\beta} = 
    \begin{cases}
      P_{\alpha,\beta}                                                &   \alpha\ne\alpha_1,\beta\ne\alpha_1 \\
      \sum_{j\in\beta} \frac1{|\alpha_1|}\sum_{i\in\alpha_1} P_{i,j}  &   \alpha=\alpha_1,\beta\ne\alpha_1    \\
      \frac1{|\alpha|} \sum_{i\in\alpha} \sum_{j\in\alpha_1} P_{i,j}  &   \alpha\ne\alpha_1,\beta=\alpha_1\\
      \frac1{|\alpha_1|} \sum_{i\in\alpha_1,j\in\alpha_1} P_{i,j}     &   \alpha=\alpha_1=\beta\,.
    \end{cases}
  \]
  Rearranging sums
  \[
    {P_1'}_{\alpha,\beta} = 
    \frac1{|\alpha|} \sum_{i\in\alpha,j\in\beta} P_{i,j} = P_{\alpha,\beta}\,,
  \]
  as expected.
\end{proof}
\begin{lemma}%
  \label{l:single-collapse}
  Collapsing a single equivalence class $\alpha$ respects $Q$ in the following sense. Let $i$ and $j$ be two non-equivalent nodes.
  \begin{itemize}
    \item If $i$ and $j$ lie in the same $\alpha$-segment, then $Q_{i,j} = Q'_{i,j}$.
    \item If $i$ and $j$ lie in different $\alpha$-segments, then $\frac{1}{2} Q_{i,j}  < Q'_{i,j} < Q_{i,j}$.
    \item If $i\in\alpha$, then $Q_{i,j} = |\alpha| Q'_{\alpha,j}$.
    \item If $j\in\alpha$, then $Q_{i,j} = Q'_{i,\alpha}$.
  \end{itemize}
\end{lemma}
\begin{proof}
  Let $\alpha=\{a_1,\ldots,a_K\}$, where $K = |\alpha|$. It is clear that $Q_{i,j}$ is unaffected by taking quotients if $i$ and $j$ lie in the same $\alpha$-segment or if $j\in\alpha$. 
  
  For $i=a_k\in\alpha$, we know that $Q_{a_k,j} = Q_{a_{\ell},j}$ for all $\ell$, so WLOG assume that $j$ lies in the segment between $a_k$ and $a_{k+1}$.
  Now, let us denote by $Q_{i_1,i_2,i_3}$ the probability of a random walker starting at $i_1$ and reaching $i_2$ before reaching $i_3$ (in particular, $Q_{i,j} = Q_{i,j,i}$). Then,
  \begin{align*}
  Q'_{\alpha,j}
  &= P'_{\alpha,j} + \sum_{i'\ne j,\alpha} P'_{\alpha,i'} Q'_{i',j,\alpha}
  = \frac1K P_{i,j} + \frac1K \sum_{\ell=1}^K \sum_{i'\ne j, i'\notin \alpha} P_{a_{\ell},i'} Q_{i',j,\alpha} \\ 
  &= \frac1K P_{i,j} + \frac1K \sum_{i'\ne j, i' \notin \alpha} P_{a_k,i'} Q_{i',j,\alpha}
  = \frac1K P_{i,j} + \frac1K \sum_{i'\ne j,i} P_{i,i'} Q_{i',j,i}
  = \frac1K Q_{i,j}\,.
  \end{align*}

  Finally let $i$ and $j$ be such that any path from $i$ to $j$ must pass through $a_1,\ldots,a_k$ before encountering $j$. Then the following reasoning applies. 
  Let $a=a_1,b=a_K$, $x = Q_{a,b,j}$ and $y = Q_{b,i,a}$. Then we have
  \begin{align*}
    Q_{i,j} &= Q_{i,a,i} Q_{a,j,i}\,, \\
    Q_{a,j,i} &= (1-x) + x Q_{b,j,i}\,, \\
    Q_{b,j,i} &= (1-y)Q_{a,j,i}\,.
  \end{align*}
  Solving for $Q_{i,j}$ yields
  \[
    Q_{i,j} = Q_{i,a}\frac{1-x}{1-x+xy}\,.
  \]
  On the quotiented graph, we also have
  \begin{align*}
    Q'_{i,j} = Q_{i,a} Q'_{\alpha,j,i}, \ \
    Q'_{\alpha,j,i} &= \frac{1}{2} \left[ 1-x + x Q'_{\alpha,j,i}\right]
    + \frac{1}{2} \left[ (1-y) Q'_{\alpha,j,i} \right]\,.
  \end{align*}
  Hence,
  $
    Q'_{i,j} = Q_{i,a}\frac{1-x}{1-x+y}
  $
  and thus $Q'_{i,j} < Q_{i,j}$. 

  Furthermore, we can bound the ratio
  \begin{equation}
    \frac{Q_{i,j}}{Q'_{i,j}} = \frac1{1-\frac{(1-x)y}{1-x+y}}\,.
    \label{rat}
  \end{equation}
  Since the function $g(x_1,x_2) = \frac{x_1x_2}{x_1+x_2}$ is bounded above by $\frac{1}{2}$ on ${(0,1)}^2$,~\cref{rat} cannot exceed $2$, which gives the bound. 
  The bound is tight because all values of $x$ and $y$ can be attained when considering arbitrary weighted graphs. (A graph with only the four nodes $a,b,i,j$ mentioned in the proof and edges $a\to j$, $a\to b$, $j\to b$, $b\to i$, $b\to a$, and $i\to a$ suffices to attain all possible values of $x,y$.) 
\end{proof}
\begin{definition}
  An \emph{absolute segment} is a maximal set of nodes which lie in the same segment with respect to all non-singleton equivalence classes.
\end{definition}
\begin{lemma}
  $\iota$ respects $Q$ in the following sense for nodes $i$ and $j$ in distinct equivalence classes $\alpha$ and $\beta$:
  \begin{itemize}
    \item If $i$ and $j$ lie in the same absolute segment, then $Q_{i,j} = Q'_{i,j}$.
    \item Otherwise, $\frac{1}{2^{c}|\alpha|} Q_{i,j} \le Q'_{\alpha,\beta} < Q_{i,j}$, where $c$ is the number of equivalence classes with respect to which $i$ and $j$ lie in different segments. (In particular, $c<L$.) Equality holds only when $c=0$.
  \end{itemize}
\end{lemma}
\begin{proof}
  If $i$ is degenerate, first collapse $\alpha$, scaling $Q_{i,j}$ by $|\alpha|$. Next, collapse all other equivalence classes one at a time, further scaling $Q_{i,j}$ by the appropriate factor in $(\frac{1}{2},1)$ whenever $i$ and $j$ lie in different segments with respect to the collapsing class.
\end{proof}
From this lemma we immediately get the following theorem.
\begin{theorem}\label{t:12bounds}
  $X'$ is a metric space with metric $(d')^{\sfrac12}$. In particular, for $i\in\alpha$ and $j\in\beta$, with $\alpha\ne\beta$:
  \begin{itemize}
    \item If $i$ and $j$ lie in the same absolute segment, then $d^{\sfrac12}_{i,j} = (d')_{i,j}^{\sfrac12}$.
    \item Otherwise $d_{i,j} < d_{\alpha,\beta}' \le d_{i,j} + \frac{1}{2} \log{|\alpha||\beta|} + c\log 2$, where $c$ is the number of equivalence classes with respect to which $i$ and $j$ lie in different segments. Equality holds only when $c=0$. 
  \end{itemize}
\end{theorem}
Thus, $\iota$ pushes apart the different absolute segments. All other distances are unaffected. 
\begin{rem}
  $\iota$ is analogous to a rigid motion on each absolute segment, in that none of the in-absolute-segment distances are distorted.
\end{rem}


\section{Computational methods}\label{s:CompMeth}

To compute the normalized hitting probabilities matrix and metric structure on a Markov chain (or network) consisting of $n$ nodes/states with probability transition matrix $P$, we require only the computation of the invariant measure and the $Q$ matrix. The invariant measure can be computed using iterative eigenvector methods, which need $O(m)$ operations per iteration for $m$ edges.

We briefly recall the work in~\cite[Theorem 5]{Thiede_2015}, that shows the $Q$ matrix can be computed in $O(n^3)$ time.  The key idea from~\cite[Lemma 5]{Thiede_2015} is that one can compute 
\begin{equation} \label{eq:Q}
  Q_{i,j} (P) = \frac{e_i^T {(I - P_j)}^{-1} P_j e_j}{e_i^T {(I - P_j)}^{-1} e_i} = \frac{ {M(j)}^{-1}_{i,j}}{ {M(j)}^{-1}_{i,i}}\,,
\end{equation}
where $e_j\in \mathbb R^n$ is the vector with a $1$ in the $j$th entry and zeros elsewhere, $P_j = (I-e_j e_j^T) P \in \mathbb R^{n\times n}$, and the invertible matrix
$M(j) = I - P + e_j e_j^T P\in \mathbb R^{n\times n}$. See Theorem $5$ of~\cite{Thiede_2015} for full details, but this identity follows from realizing that as defined $M(j)$ is invertible with inverse
\[
M(j)^{-1} = \begin{pmatrix}
(I-P_j)^{-1} & (I-P_j)^{-1} P_j e_j \\
0 & 1
\end{pmatrix}
\]
given in block form on the $e_j^\perp$, $e_j$ basis.

If we then compute $M(1)^{-1}$ on the way to obtaining the first column $Q_{i,1} = {M(1)}^{-1}_{i1}/{M(1)}^{-1}_{ii}$,
then $M(j)$ is a rank-$2$ perturbation of $M(1)$ and we can apply the Sherman-Morrison-Woodbury identity to compute $M(j)^{-1}$. Since we only access $2n-2$ elements of $M(j)^{-1}$, the full $O(n^2)$-time Sherman-Morrison-Woodbury update is not needed, and we can get the $j$th column $Q_{i,j}$ in $O(n)$ computations from ${M(1)}^{-1}$. A \textsc{MATLAB} implementation of this procedure, along with code for all of the numerical experiments described in the paper, is available at~\url{https://github.com/zboyd2/hitting_probabilities_metric}.

The matrix $Q$ encodes the hitting probabilities of a random walk on the nodes of a graph and the order of the method we present here is very well documented in \cite{Thiede_2015}.  However, there are several results that consider the computational complexity of the related problem of commute times, see for instance the works \cite{li2010random,boley2011commute}.  The computational cost of computing hitting probabilities through inversion of the Laplacian has been explored further in \cite{golnari2019markov,cohen2016faster,cohen2018solving}, resulting in some cases in which the method may be improved to better than $O(n^3)$. As we are mostly interested in the construction of the metric here, we will not further explore the question of optimal order of the computation.

\section{Examples}\label{s:Examples}%
We consider examples of Markov Chains and directed graphs to illustrate the proposed metric. We start with simple graphs for which the calculations can be performed exactly. We then numerically explore a variety of synthetic graphs and a real-world example defined from New York City taxi cab data.

\subsection{Exact formulations}%
\label{sec:glued}

Here we consider some simple graphs on which the invariant measure and hitting probabilities can be computed exactly to help us understand $\Ahtb$ and $d^{\beta}$.

\begin{enumerate}
    \item \emph{Directed cycles:} Consider a directed cycle on $n$ nodes. Then $\phi_i = \sfrac{1}{n}$ for all $i$, and $Q_{i,j} = 1$ for all $i\ne j$. Therefore, $\Ahtb$ is a weighted clique, and $d^{\beta}$ has all points equidistant. For $\beta=\sfrac12$, the weights equal to $1$ and all nodes are identified with each other in the metric topology.

\item \emph{Complete graphs:} Consider a complete graph on $n>2$ nodes. Then $\phi_i = \sfrac{1}{n}$ for all $i$, and $Q_{i,j} = \const < 1$ for all $i \ne j$. Therefore, $\Ahtb$ is a weighted clique. Unlike the directed cycle case, the weights in the clique are $<1$ for all  $\beta\geq\sfrac12$.

\item \emph{Glued cycles:}
Consider graphs of the type depicted in~\cref{fig:glue}, namely graphs composed of $n_b$ ``backbone nodes'' forming a directed chain, which then branches into $C$ chains of length $n_c$, each of which finally connects back to the beginning of the backbone chain. Intuitively, a random walker on this graph transitions between $C+1$ groups of nodes, namely, each of the $C$ branches and the backbone. As illustrated in~\cref{tab:glue}, our metric captures this intuition by placing each node very close to the others on its chain. This is in contrast to commute-time-based metrics, where the length of the chain must be taken into account. (See~\cref{fig:effRes}.)
In~\cref{sec:num_ex}, we consider some numerical results based on this example.
\end{enumerate}

\begin{table}
    \centering
    \begin{tabular}{llcccc}
        \toprule
    $i$       &   $j$           &   $Q$         &   $\frac{\Ahtb[\beta]}{(n_b + n_c)^{1-2\beta}}$  &  $\Ahtb[\sfrac12]$  &   $d^\frac12$ \\
        \midrule
    branch    & same branch      & $1$          &   $\sfrac1{C^{2\beta-1}}$                        &  $1$                &   $0$                \\\addlinespace[2pt]
    branch    & different branch & $\sfrac12$   &   $\sfrac1{2 C^{2\beta-1}}$                      &  $\sfrac12$         &   $\log 2$         \\\addlinespace[2pt]
    backbone  & branch           & $\sfrac1C$   &   $\sfrac1{C^{\beta}}$                           &  $C^{-\sfrac12}$    &   $\sfrac12 \log C$  \\\addlinespace[2pt]
    branch    & backbone         & $1$          &   $\sfrac1{C^{\beta}}$                           &  $C^{-\sfrac12}$    &   $\sfrac12 \log C$  \\\addlinespace[2pt]
    backbone  & backbone         & $1$          &   $1$                                            &  $1$                &   $0$                \\\addlinespace[2pt]
        \bottomrule
    \end{tabular}
    \caption{Values of $Q$, $\Ahtb$, and $d$ evaluated at distinct nodes $i$ and $j$ for the glued cycles example from~\cref{sec:glued}. We include extra columns for the case $\beta=\sfrac12$, which is particularly interpretable. Observe that neither $\Ahtb$ nor $d^{\beta}$ depends on $n_b$ or $n_t$ (except up to scaling), which is a manifestation of their blindness to walk length. Also, the nodes that are closest together are those which lie on common chains. Note that we scaled $\Ahtb$ for visual clarity. The invariant measure is easily verified to be $\left(n_b + n_c\right)^{-1}$ on the backbone and $\left(C (n_b + n_c)\right)^{-1}$ elsewhere.}%
    \label{tab:glue}
\end{table}

\subsection{Synthetic numerical examples}%
\label{sec:num_ex}
We consider four examples. The first two demonstrate that the spectrum of $\Ahtb$ (for $\beta = \sfrac12$ or $\beta=1$) identifies cyclic and clique-like sets in a useful manner. We compare to two alternative symmetrizations and another metric. The second example additionally shows the scalability of our approach. In the third example, we explore when it is advantageous to use $d$ for visualization and clustering purposes, using a directed planted partition model for ground truth comparisons. In dense, difficult-to-detect regimes, our method is more accurate than clustering using the input adjacency matrix directly. Finally, in the fourth example, we compare $d^{\sfrac12}$, $d^{1}$, and spatial distance for geometric graphs, finding that our distance captures comparable information to the spatial distance, with the similarity being especially tight when $\beta=\sfrac12$.

\subsubsection{Glued-cycles networks}%
\label{sec:gluedcycles}

For the two-glued-cycles networks illustrated in~\cref{fig:glue}, we construct a probability transition matrix $P$ by taking a uniform edge weight for all connected vertices and performing a row normalization.  We then compute the Fiedler eigenvector corresponding to the second smallest eigenvalue of the graph Laplacian (sometimes called the Fiedler vector) for different symmetrized adjacency matrices.  For the adjacency matrices $A$ constructed below, we calculate the graph Laplacian $L=D-A$, with $D$ the diagonal matrix of node degrees (row or column sums of $A$). The examples here are two directly glued cycles, as well as two glued cycles with a bidirectional edge between the cycles.  In the first case, the results are all very similar regardless of the symmetrization, but for the second case the results differ significantly.  In each case, we group the nodes based on whether the corresponding vector element is positive, negative, or zero. 

\begin{figure}
\centering
\begin{tikzpicture}[scale=0.75] 
\begin{scope}[every node/.style={circle,thick,draw}]
    \node (A) at (0,0) {};
    \node (B) at (0,1) {};
    \node (C) at (0,2) {};
    \node (D) at (1,2.5) {};
    \node (E) at (1.5,1.5) {};
    \node (F) at (1.5,.5) {} ;
    \node (G) at (1,-.5) {} ;
    \node (H) at (-1,2.5) {} ;
    \node (I) at (-1.5,1.5) {} ;
    \node (J) at (-1.5,.5) {} ;
    \node (K) at (-1,-.5) {} ;
\end{scope}

\begin{scope}[>={Stealth[black]},
              every node/.style={fill=white,circle},
              every edge/.style={draw=red,very thick}]
    \path [->] (A) edge (B);
    \path [->] (B) edge (C);
    \path [->] (C) edge (D);
    \path [->] (D) edge (E);
    \path [->] (E) edge (F);
    \path [->] (F) edge (G);
    \path [->] (G) edge (A);
    \path [->] (C) edge (H);
    \path [->] (H) edge (I); 
    \path [->] (I) edge (J); 
    \path [->] (J) edge (K); 
    \path [->] (K) edge (A); 
\end{scope}
\end{tikzpicture}  \ \ \
\begin{tikzpicture}[scale=0.75] 
\begin{scope}[every node/.style={circle,thick,draw}]
    \node (A) at (0,0) {};
    \node (B) at (0,1) {};
    \node (C) at (0,2) {};
    \node (D) at (1,2.5) {};
    \node (E) at (1.5,1.5) {};
    \node (F) at (1.5,.5) {} ;
    \node (G) at (1,-.5) {} ;
    \node (H) at (-1,2.5) {} ;
    \node (I) at (-1.5,1.5) {} ;
    \node (J) at (-1.5,.5) {} ;
    \node (K) at (-1,-.5) {} ;
\end{scope}

\begin{scope}[>={Stealth[black]},
              every node/.style={fill=white,circle},
              every edge/.style={draw=red,very thick}]
    \path [<->] (F) edge (J); 
    \path [->] (A) edge (B);
    \path [->] (B) edge (C);
    \path [->] (C) edge (D);
    \path [->] (D) edge (E);
    \path [->] (E) edge (F);
    \path [->] (F) edge (G);
    \path [->] (G) edge (A);
    \path [->] (C) edge (H);
    \path [->] (H) edge (I); 
    \path [->] (I) edge (J); 
    \path [->] (J) edge (K); 
    \path [->] (K) edge (A); 
\end{scope}
\end{tikzpicture} \ \ \
\begin{tikzpicture}[scale=0.75] 
\begin{scope}[every node/.style={circle,thick,draw}]
    \node[fill=magenta] (A) at (0,0) {};
    \node[fill=magenta] (B) at (0,1) {};
    \node[fill=magenta] (C) at (0,2) {};
    \node[fill=blue] (D) at (1,2.5) {};
    \node[fill=blue] (E) at (1.5,1.5) {};
    \node[fill=blue] (F) at (1.5,.5) {} ;
    \node[fill=blue] (G) at (1,-.5) {} ;
    \node[fill=green] (H) at (-1,2.5) {} ;
    \node[fill=green] (I) at (-1.5,1.5) {} ;
    \node[fill=green] (J) at (-1.5,.5) {} ;
    \node[fill=green] (K) at (-1,-.5) {} ;
\end{scope}

\begin{scope}[>={Stealth[black]},
              every node/.style={fill=white,circle},
              every edge/.style={draw=red,very thick}]
    \path [->] (A) edge (B);
    \path [->] (B) edge (C);
    \path [->] (C) edge (D);
    \path [->] (D) edge (E);
    \path [->] (E) edge (F);
    \path [->] (F) edge (G);
    \path [->] (G) edge (A);
    \path [->] (C) edge (H);
    \path [->] (H) edge (I); 
    \path [->] (I) edge (J); 
    \path [->] (J) edge (K); 
    \path [->] (K) edge (A); 
\end{scope}
\end{tikzpicture} \\ 
\vspace{.2cm}
\begin{tabular}{cccc}
\begin{tikzpicture}[scale=0.75] 
\begin{scope}[every node/.style={circle,thick,draw}]
    \node[fill=blue] (A) at (0,0) {};
    \node[fill=blue] (B) at (0,1) {};
    \node[fill=green] (C) at (0,2) {};
    \node[fill=green] (D) at (1,2.5) {};
    \node[fill=green] (E) at (1.5,1.5) {};
    \node[fill=blue] (F) at (1.5,.5) {} ;
    \node[fill=blue] (G) at (1,-.5) {} ;
    \node[fill=green] (H) at (-1,2.5) {} ;
    \node[fill=green] (I) at (-1.5,1.5) {} ;
    \node[fill=blue] (J) at (-1.5,.5) {} ;
    \node[fill=blue] (K) at (-1,-.5) {} ;
\end{scope}
\begin{scope}[>={Stealth[black]},
              every node/.style={fill=white,circle},
              every edge/.style={draw=red,very thick}]
    \path [<->] (F) edge (J); 
    \path [->] (A) edge (B);
    \path [->] (B) edge (C);
    \path [->] (C) edge (D);
    \path [->] (D) edge (E);
    \path [->] (E) edge (F);
    \path [->] (F) edge (G);
    \path [->] (G) edge (A);
    \path [->] (C) edge (H);
    \path [->] (H) edge (I); 
    \path [->] (I) edge (J); 
    \path [->] (J) edge (K); 
    \path [->] (K) edge (A); 
\end{scope}
\end{tikzpicture}
&
\begin{tikzpicture}[scale=0.75] 
\begin{scope}[every node/.style={circle,thick,draw}]
    \node[fill=green] (K) at (-1,-.5) {} ;
    \node[fill=green] (A) at (0,0) {};
    \node[fill=green] (B) at (0,1) {};
    \node[fill=green] (C) at (0,2) {};
    \node[fill=blue] (D) at (1,2.5) {};
    \node[fill=blue] (E) at (1.5,1.5) {};
    \node[fill=blue] (F) at (1.5,.5) {} ;
    \node[fill=green] (G) at (1,-.5) {} ;
    \node[fill=blue] (H) at (-1,2.5) {} ;
    \node[fill=blue] (I) at (-1.5,1.5) {} ;
    \node[fill=blue] (J) at (-1.5,.5) {} ;
\end{scope}
\begin{scope}[>={Stealth[black]},
              every node/.style={fill=white,circle},
              every edge/.style={draw=red,very thick}]
    \path [<->] (F) edge (J); 
    \path [->] (A) edge (B);
    \path [->] (B) edge (C);
    \path [->] (C) edge (D);
    \path [->] (D) edge (E);
    \path [->] (E) edge (F);
    \path [->] (F) edge (G);
    \path [->] (G) edge (A);
    \path [->] (C) edge (H);
    \path [->] (H) edge (I); 
    \path [->] (I) edge (J); 
    \path [->] (J) edge (K); 
    \path [->] (K) edge (A); 
\end{scope}
\end{tikzpicture}
&
\begin{tikzpicture}[scale=0.75] 
\begin{scope}[every node/.style={circle,thick,draw}]
    \node[fill=magenta] (A) at (0,0) {};
    \node[fill=magenta] (B) at (0,1) {};
    \node[fill=magenta] (C) at (0,2) {};
    \node[fill=blue] (D) at (1,2.5) {};
    \node[fill=blue] (E) at (1.5,1.5) {};
    \node[fill=blue] (F) at (1.5,.5) {} ;
    \node[fill=blue] (G) at (1,-.5) {} ;
    \node[fill=green] (H) at (-1,2.5) {} ;
    \node[fill=green] (I) at (-1.5,1.5) {} ;
    \node[fill=green] (J) at (-1.5,.5) {} ;
    \node[fill=green] (K) at (-1,-.5) {} ;
\end{scope}
\begin{scope}[>={Stealth[black]},
              every node/.style={fill=white,circle},
              every edge/.style={draw=red,very thick}]
    \path [<->] (F) edge (J); 
    \path [->] (A) edge (B);
    \path [->] (B) edge (C);
    \path [->] (C) edge (D);
    \path [->] (D) edge (E);
    \path [->] (E) edge (F);
    \path [->] (F) edge (G);
    \path [->] (G) edge (A);
    \path [->] (C) edge (H);
    \path [->] (H) edge (I); 
    \path [->] (I) edge (J); 
    \path [->] (J) edge (K); 
    \path [->] (K) edge (A); 
\end{scope}
\end{tikzpicture}
&
\begin{tikzpicture}[scale=0.75] 
\begin{scope}[every node/.style={circle,thick,draw}]
    \node[fill=magenta] (A) at (0,0) {};
    \node[fill=magenta] (B) at (0,1) {};
    \node[fill=magenta] (C) at (0,2) {};
    \node[fill=blue] (D) at (1,2.5) {};
    \node[fill=blue] (E) at (1.5,1.5) {};
    \node[fill=blue] (F) at (1.5,.5) {} ;
    \node[fill=blue] (G) at (1,-.5) {} ;
    \node[fill=green] (H) at (-1,2.5) {} ;
    \node[fill=green] (I) at (-1.5,1.5) {} ;
    \node[fill=green] (J) at (-1.5,.5) {} ;
    \node[fill=green] (K) at (-1,-.5) {} ;
\end{scope}

\begin{scope}[>={Stealth[black]},
              every node/.style={fill=white,circle},
              every edge/.style={draw=red,very thick}]
    \path [<->] (F) edge (J); 
    \path [->] (A) edge (B);
    \path [->] (B) edge (C);
    \path [->] (C) edge (D);
    \path [->] (D) edge (E);
    \path [->] (E) edge (F);
    \path [->] (F) edge (G);
    \path [->] (G) edge (A);
    \path [->] (C) edge (H);
    \path [->] (H) edge (I); 
    \path [->] (I) edge (J); 
    \path [->] (J) edge (K); 
    \path [->] (K) edge (A); 
\end{scope}
\end{tikzpicture} \\
$A = \max(P,P^T)$    &   Chung's $L$~\cite{Chung_2005}  &   $A =\Ahtb[1]$  &   $A=\Ahtb[\sfrac12]$
\end{tabular}
  \caption{(Top Left) Two-glued-cycles example from~\cref{sec:glued} with $n_b=3$, $n_c=4$, and $C=2$. The ``backbone'' nodes run along the center, and the two partial cycles split off from and then return to it.  (Top Middle) Similar two-glued-cycles network with a bidirectional edge.  (Top Right) Sign of the Fiedler vector of the Laplacian for several different symmetrizations. (Bottom) The sign of the Fiedler vector of the Laplacian for several different symmetrizations.
  The sign of the Fiedler vectors is encoded as ($-$, green), ($0$, magenta) and ($+$, blue).}
  \label{fig:glue}
\end{figure}
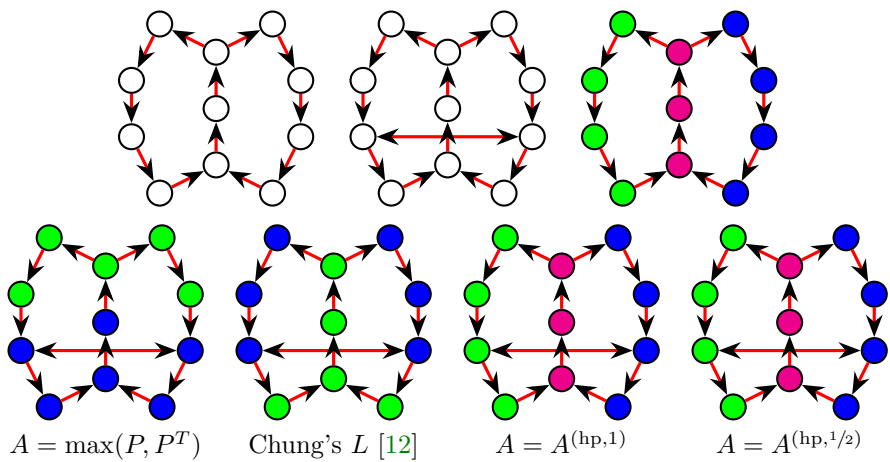
For the two glued cycles without the bidirectional edge, the naive symmetrizations of the directed adjacency matrix, either  $A = (P+P^T)/2$ or $A=\max\left(P,P^T\right)$, 
have a Fiedler vector that is $0$ on the spine and splits each cycle into signed components, see the top right plot in~\cref{fig:glue}.  However, in the bottom left component~\cref{fig:glue}, for the two-glued-cycles network with the bidirectional edge, the naive symmetrization splits the network horizontally, which is reasonable, since the resulting graph cut is small, although this (by construction) does not reflect the coherent, directed structure of the original graph. 

One way to account for directed structure in a way that minimizes equilibrium flux across the cut was suggested by Fan Chung~\cite{Chung_2005} (cited in \cref{s:RelWork}), defining the Laplacian by $
  L = I - \frac{1}{2} \left[ \Phi^{\sfrac 1 2} P \Phi^{-\sfrac 1 2} + \Phi^{-\sfrac 1 2} P^T \Phi^{\sfrac 1 2} \right], $
where $\Phi = \textrm{diag}(\phi) \in \mathbb R^{n \times n}$. 
Chung uses $L$ to establish a Cheeger-type inequality for digraphs, which is used to study the rate of convergence for Markov chains. Using Chung's Laplacian again gives a comparable outcome for the two glued cycles example (\cref{fig:glue}), but in the example with the bidirectional edge, this symmetrization places most of the non-backbone nodes in one class and all backbone nodes in the other (second plot in~\cref{fig:glue}).

The normalized hitting probabilities matrices $\Ahtb[1]$ and $\Ahtb[\sfrac12]$ each distinguish between the two branches, with the backbone set equal to zero in both the cases of the glued cycles and the glued cycles with a bidirectional edge as seen in~\cref{fig:glue}.
Thus, all three approaches uncover different structure in the two-glued-cycle graph with the bidirectional edge, with the naive symmetrization yielding small undirected cuts, Chung's approach yielding (perhaps) two different dynamical states, and $\Ahtb[\beta]$ showing all three chains in a natural way for both $\beta = \sfrac12, 1$.

Finally, we compare the total effective resistance metric of~\cite{Young_2016b} to our metric on the example of the two-glued-cycles network (with no bidirectional edge). As one might expect given the relation ship between effective resistance and commute times in the undirected case, the total effective resistance of~\cite{Young_2016b} is sensitive to cycle length.  
\Cref{fig:effRes} demonstrates that the commute time approach views the distances on each cycle quite differently and that the relative distances from the total effective resistance metric are more difficult to interpret in the second loop.

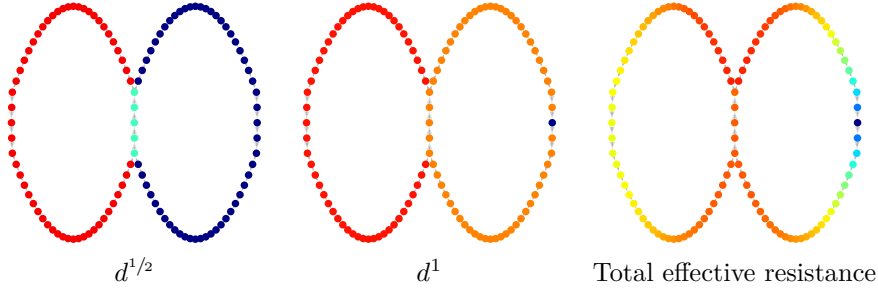
\begin{figure}
  \centering
  \begin{tabular}{ccc}
    \begin{tikzpicture}[xscale=.35,yscale=.4]
      \Vertices[RGB=true,size=.1]{nodesd12.dat}
      \Edges[Direct=True,lw=.01,opacity=.2]{edges.dat}
    \end{tikzpicture}
    &
    \begin{tikzpicture}[xscale=.35,yscale=.4]
      \Vertices[RGB=true,size=.1]{nodesd1.dat}
      \Edges[Direct=True,lw=.01,opacity=.2]{edges.dat}
    \end{tikzpicture}
    &
    \begin{tikzpicture}[xscale=.35,yscale=.4]
      \Vertices[RGB=true,size=.1]{nodesR.dat}
      \Edges[Direct=True,lw=.01,opacity=.2]{edges.dat}
    \end{tikzpicture}
    \\
    $d^{\sfrac12}$ &
    $d^1$ &
    Total effective resistance
  \end{tabular}
  \caption{Two glued cycles with $n_c=55$ and $n_b=5$, with nodes colored by distance from a node on the far right. Blue denotes small distances. The metric $d^{\sfrac12}$ has three levels of distance corresponding to nodes on the same branch, backbone, and opposite branch, respectively. The metric $d^1$ is similar, except nodes on the same branch are not distinguished from backbone nodes. Finally, the total effective resistance metric from~\cite{Young_2016b} gives a smoother notion of distance on the right branch and backbone, but on the left branch,  proceeding counterclockwise, one finds the distance decreasing and then increasing again, which is somewhat difficult to interpret. This example shows how different resistance/commute time are from hitting-probability distance.}%
  \label{fig:effRes}
\end{figure}


\subsubsection{Cycle adjoined to directed Erd\H{o}s--R\'enyi Graph}
Consider the following construction, illustrated in~\cref{fig:er}.
Let $n=n_{\mathrm{er}} + n_{\mathrm{cycle}}$, and let the
\begin{wrapfigure}[15]{r}{.35\textwidth}
  \centering
  \includegraphics[width=.35\textwidth]{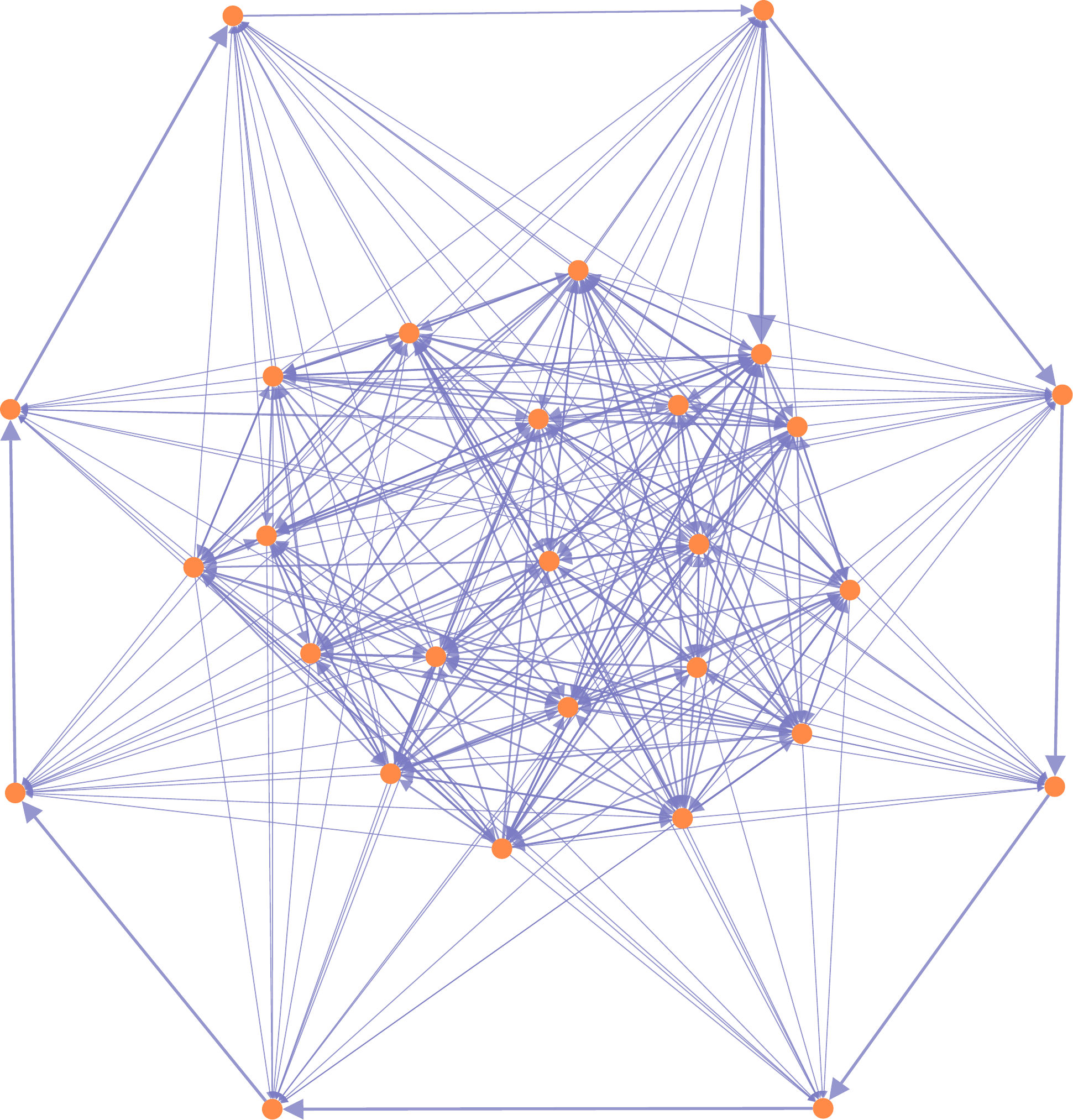}
  \vspace*{-0.2in}
  \caption{Erd\H{o}s-R\'enyi plus cycle example.}
  \label{fig:er}
\end{wrapfigure}
first $n_{\mathrm{er}}$ nodes form an unweighted, directed ER graph with connection probability $p$. 
The remaining $n_{\mathrm{cycle}}$ nodes form an unweighted, directed cycle. 
An adjacency matrix for the ER graph and cycle are connected by adding $2\, \mathrm{round}(n p)-1$ edges of weight $w$ to each cycle node from randomly selected nodes in the ER graph.\footnote{These edges are drawn with replacement with multi-edges merged to a single edge of weight $w$. Results were similar when we added the weights instead.} Finally, a single, bidirectional edge of weight $1$ is added from one cycle node to one ER node. 
Normalizing the rows to form a probability transition matrix, a random walker on this graph would transition between the ER and cycle subgraphs, 
where the cycle subgraph is difficult to escape quickly because of the single exit. 
For the particular choice of $n_{\mathrm{er}}=20$, $n_{\mathrm{cycle}}=8$, $p=.5$, and $w=3$, we find that the Fiedler vector of (the Laplacian associated with) $\Ahtb[\sfrac12]$ is positive on the cycle nodes and negative elsewhere. In contrast, the Fiedler vector of the naive symmetrization $A=(P+P^T)/2$ or Chung's $L$~\cite{Chung_2005} does not separate the cycle and ER nodes.
Scaling up to $n=n_{\mathrm{er}}+n_{\mathrm{cycle}} = 7,200 + 2,800 = 10,000$ nodes keeping the other parameters the same ($\approx 38.7$ million edges) gives similar eigenvector results. 
The computation takes 31 seconds on a Lenovo ThinkStation P410 desktop with Xeon E5--1620V4 3.5 GHz CPU and 16 GB RAM using {\sc MATLAB} R2019a Update 4 (9.6.0.1150989) 64-bit (glnxa64): 18 seconds to compute $Q$, 6 seconds to compute $\phi$, 2 seconds to form $\Ahtb[\sfrac{1}{2}]$, and 5 seconds to compute the Fiedler vector.

\subsubsection{Cluster detection and visualization for digraphs}\label{sec:kmeans}

We next use $d$ for clustering and dimension reduction. 
We consider directed graphs generated by a planted partition model with nodes grouped into three ground truth communities and form a uniformly weighted adjacency matrix by connecting an edge from $i$ to $j$ with probability $p_{\mathrm{in}}$ if $i$ and $j$ are in the same community and $p_{\mathrm{out}}$ ($<p_{\mathrm{in}}$) otherwise. A probability transition matrix can then be formed using row normalization.
We then attempt to recover the ground truth node assignments. 
\begin{wrapfigure}[15]{r}{.32\textwidth}
  \centering
  \includegraphics[width=.3\textwidth]{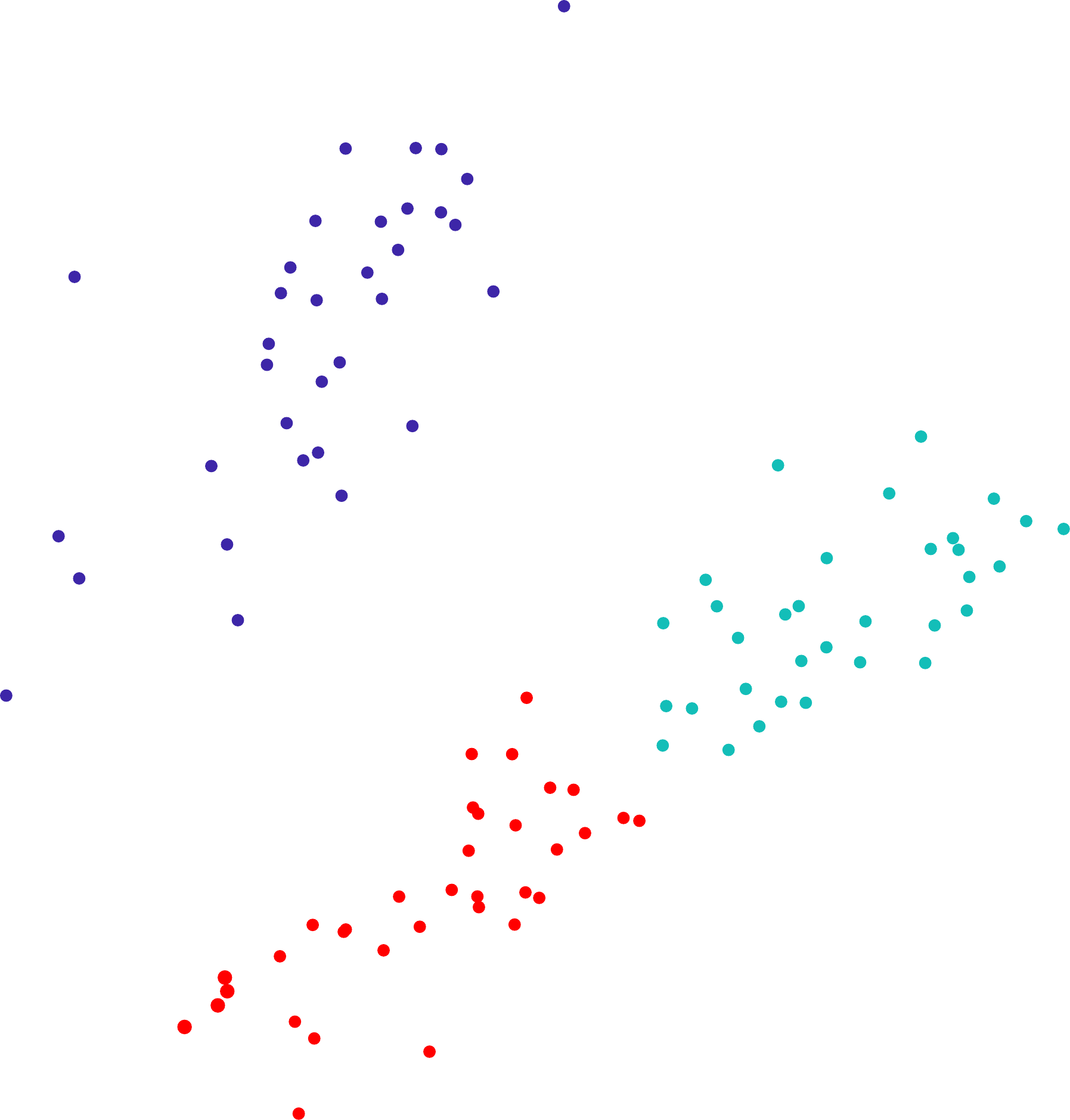}
  \caption{PCA embedding of $d^{\sfrac12}$, colored by ground truth community.}
  \label{fig:pca}
\end{wrapfigure}
The difficulty of this problem is generally understood in terms of $\Delta = p_{\mathrm{in}} - p_{\mathrm{out}}$ and $\rho = \frac{p_{\mathrm{in}} + 2 p_{\mathrm{out}}}{3}$. Small values of $\Delta$ correspond to more difficult clustering problems that may be solved less accurately (relative to the ground truth). 
In this example we attempt to cluster the nodes into $k=3$ clusters using several approaches: (1) principal component analysis\footnotemark (PCA)~\cite{pca} on the adjacency matrix, $A$, followed by $k$-means clustering on the first $k-1$ PCA vectors; (2) PCA on $d^{\sfrac12}$ followed by $k$-means; and (3) $k$-medoids on $d^{\sfrac12}$. 
(The $k$-medoids algorithm is similar in spirit to the $k$-means unsupervised clustering algorithm but applies in arbitrary metric spaces, see for instance \cite{kaufmann1987clustering,park2009simple}.) Results are shown in~\cref{fig:kmedoids,fig:regions}.
\footnotetext{Specifically, we used the PCA routine from MATLAB R2019a Update 4 (9.6.0.1150989) 64-bit (glnxa64). As expected, this gives different results in general when applied to a matrix versus its transpose. In this case, the matrix is stochastically equivalent with its transpose, and in the NY taxi example below, the PCA-based plots are similar regardless of whether the transpose is used.}
\begin{figure}[!ht]
  \centering
  \footnotesize
    \def\svgwidth{\linewidth}
    \input{clusteringstudyinput.tex}
  \caption{Results of (top row) PCA on $A$ followed by $k$-means, (middle) PCA on $d^{\sfrac12}$ followed by $k$-means, and (bottom) $k$-medoids on $d^{\sfrac12}$ on 300-node graphs generated using the directed planted partition model with three clusters, as described in~\cref{sec:kmeans}. We varied the mean edge density, $\rho$, and cluster quality, $\Delta=p_{\mathrm{in}}-p_{\mathrm{out}}$. Since results depend on the random initialization, we report best of 5 runs for each entry. If any generated graph was not strongly connected, we did not try to cluster it. The left column is the accuracy (purity) of the recovered partition, and the right value is the empirical $p$ value of the accuracy relative to 4,000 random partitions obtained by drawing each community label uniformly at random. Notably, method (2) has the best performance for dense, weakly-clustered graphs. [Note that the triangular blue region on the lower left of each plot represents a $(\rho,\Delta)$ parameter combination that cannot exist.]}%
  \label{fig:kmedoids}
\end{figure}
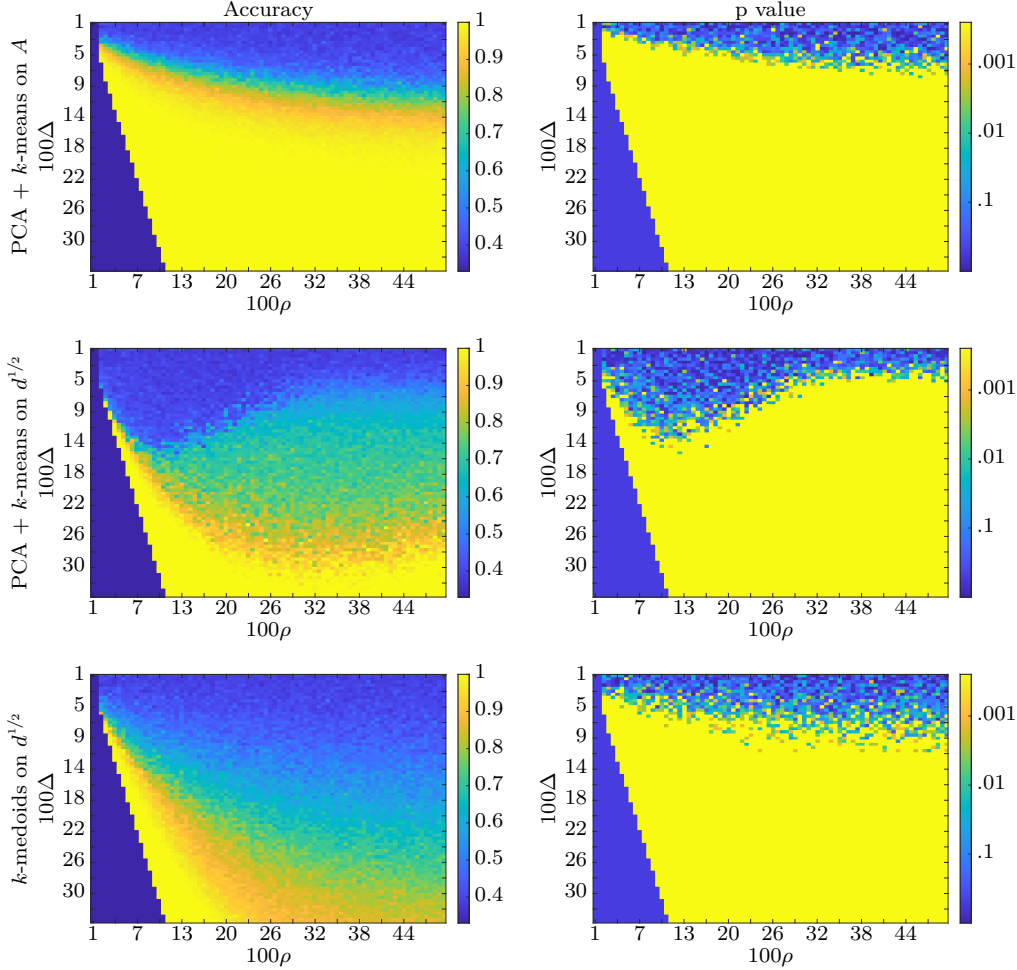
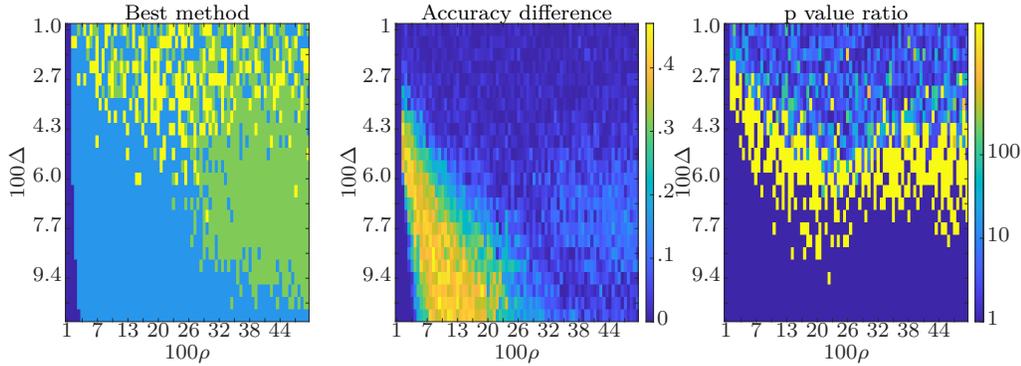
\begin{figure}
  \centering
  \footnotesize
    \def\svgwidth{\linewidth}
    \input{clusteringstudy1input.tex}
  \caption{(Left) Regions where the methods from~\cref{fig:kmedoids} perform best. Here, light blue is method (1), green is method (2), and yellow is method (3). (Middle) Difference in accuracy between the best and second-best methods. (Right) Ratio of $p$ value of the best and second best methods. from~\cref{fig:kmedoids}. Combining these plots, we see that there is a significant parameter regime consisting of dense, difficult to detect structure, where using $d^{\sfrac12}$ instead of $A$ enhances the spectral detection of structure by 5--20\% for graphs where method (1) is recovering essentially no structure in $A$. Note that the y axis is different from~\cref{fig:kmedoids}.}%
  \label{fig:regions}
\end{figure}
We find that 
method (1) works best on sparse or well-separated clusters, method (2) works best with dense, difficult-to-detect clusters, and method (3) has no clear advantage. More specifically, using $d^{\sfrac12}$ in method (2) enhances our ability to get a better-than-chance clustering in dense networks.\footnote{We also tried using the shortest commute and generalized effective resistance metrics~\cite{Young_2016b,Young_2016a} as substitute for the hitting time metric in this example and found similar improvements over using the raw adjacency matrix. In particular, the shortest commute was the most effective metric for this task (although this metric is not robust, so the real-life performance may be different).} (We note that spectral methods in undirected graphs give asymptotically optimal almost-exact recovery but are not optimal for harder cases where only better-than-chance recoverability is possible~\cite{abbe}. This is consistent with~\cref{fig:kmedoids,fig:regions}.) 
Finally, we can also use PCA on $d^{\sfrac12}$ to visualize the directed network. The first and second principal components, generated using {\sc MATLAB}'s built-in routine, are plotted in~\cref{fig:pca}, clearly showing the separation into three clusters, which are in accordance with the three ground-truth communities.

\subsubsection{Distances on geometric graphs} 

Given known convergence properties of various graph models to continuum problems (e.g. \cite{trillos2016consistency,trillos2018variational,singer2012vector,singer2017spectral,Osting_2017}), we are motivated by the question of how our distance metric compares to a standard notion of distance when the network arises from a natural geometric setting.  For instance, as mentioned in the introduction, \cite{singer2012vector} proves that the notion of diffusion distance converges to that of geodesic distance as a point cloud samples a closed manifold at higher and higher densities.

In~\cref{fig:geodistance}, we consider distances computed using our metric structure in a family of geometric graphs constructed using Euclidean distances to determine edge weights.  The geometric graphs considered are 

\begin{enumerate}[(a)]
  \item A random point cloud on a flat torus ${[0,2 \pi]}^2$ with $36^2$ points,
  \item A random point cloud on a flat torus with a hole ${[0,2 \pi]}^2 \setminus B((\pi,\pi),\pi/2)$ with $36^2$ points (distances relative to a point in the bottom left of the torus),
  \item  An $H$ shaped domain $({[0,2 \pi]}^2 \cap \{|x_1-\pi|\geq \pi/2\}) \cup ({[0,2 \pi]}^2 \cap \{|x_2-\pi|\leq \pi/4\})$ with $36^2$ points (distances relative to a point in the bottom right of the $H$), 
  \item A random point cloud on the circle of length $2 \pi$ with $1000$ points, 
  \item A random point cloud on a sphere of radius $1$ in $\mathbb{R}^3$ with $1000$ points,
  \item  A square $10 \times 10$ lattice on the flat torus ${[0,2 \pi]}^2$.
\end{enumerate} 
For the regular lattice example, the edge weights are only carried on nearest neighbor vertices.  In all other cases, we consider the edge weights to be of the form $e^{-\gamma d_{\text{Euc}}{(x_i,x_j)}^2}$, where $d_{\text{Euc}}$ is just the Euclidean distance metric (determined with periodicity if the domain is periodic, \ie, we take shortest path distance in the flat torus).  We have chosen the scale factor $\gamma = 1$ uniformly throughout.

Once the geometric graph is constructed, we computed the pairwise Euclidean distances, as well as the pairwise distances $d^{\sfrac12}$ and $d^1$ for comparison.  To assist with interpretation and comparison, we have ordered the vertices in~\cref{fig:geodistance} from closest to farthest relative to the $d^{\sfrac12}$ metric and plotted for each distance function the rescaled distances $(d-d_{\min})/(d_{\max} - d_{\min})$ to normalize all of them to the same scale.

Throughout, we note that $d^{\sfrac12}$ is a reasonable fit to the measured Euclidean distances, while $d^1$ seems to do well only when the geometry is such that the invariant measure normalization (that is, the choice of $\beta$) does not matter as much. Note that the distance $d^{\sfrac12}$ and $d^1$ are identical on the square lattice, up to scaling. 
In this case, we are really studying the structure of the $Q$ hitting probability matrix.  Our results give some preliminary indication that in the consistency limit the $d^{\sfrac12}$ metric may converge to the Euclidean distance while the $d^1$ metric converges to something else entirely.  However, we leave this pursuit for future analytical studies.

\begin{figure}[t!]
  \centering
  \begin{subfigure}[t]{0.3\textwidth}
    \centering
    \includegraphics[width=\textwidth]{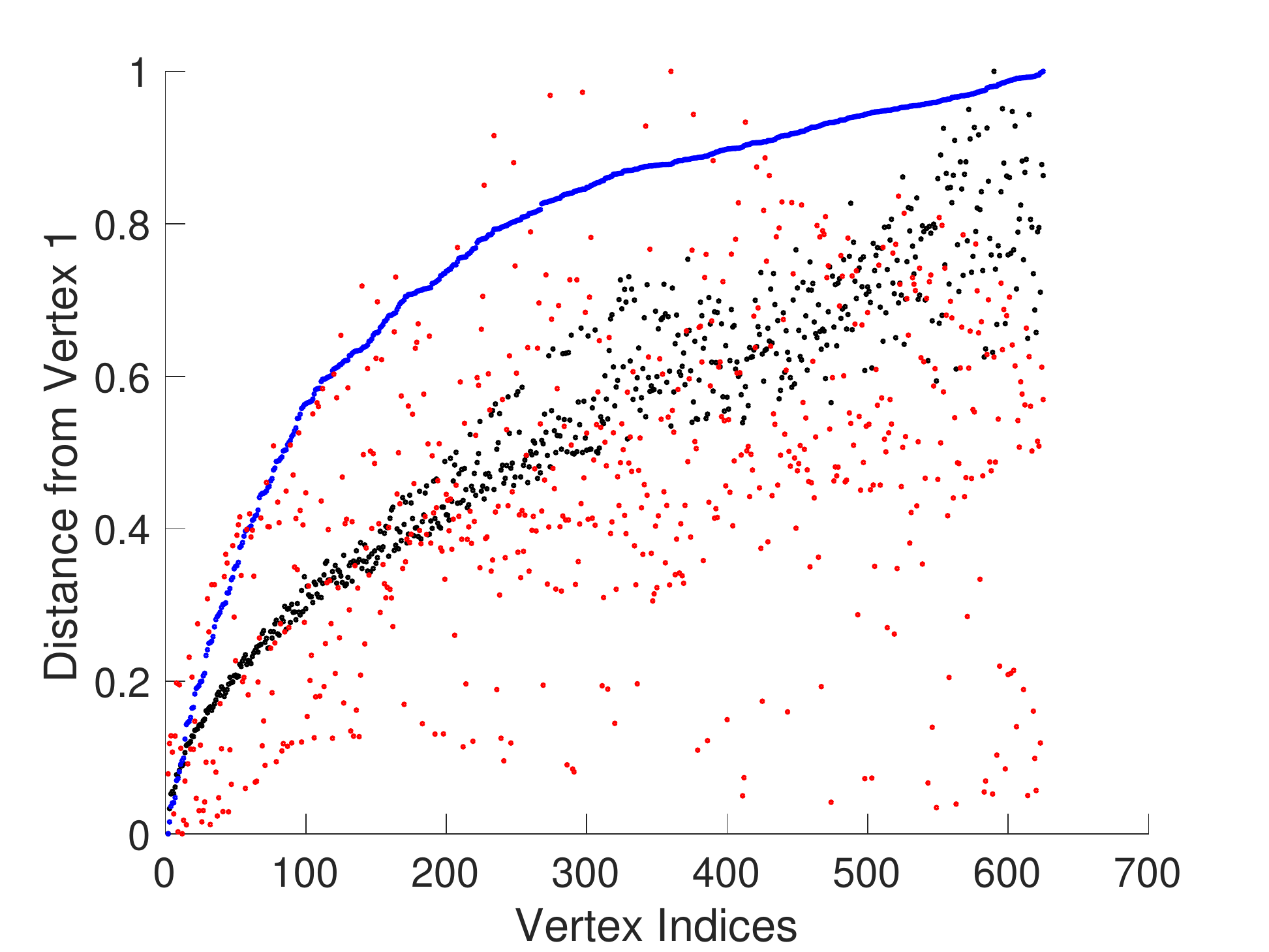}
    \caption{Flat Torus}
  \end{subfigure}\quad
  \begin{subfigure}[t]{0.3\textwidth}
    \centering
    \includegraphics[width=\textwidth]{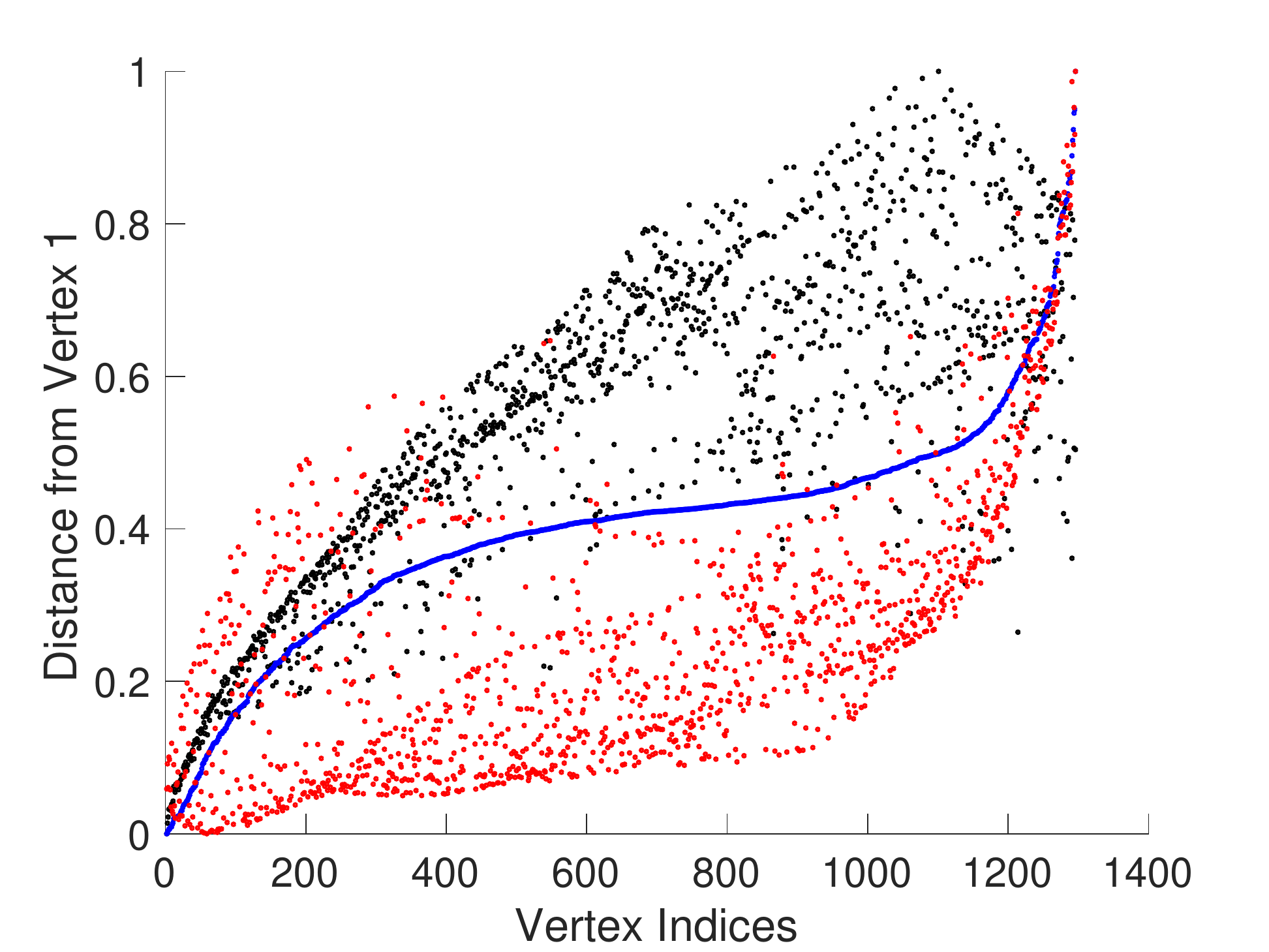}
    \caption{Flat Torus with a Hole}
  \end{subfigure}\quad
  \begin{subfigure}[t]{0.3\textwidth}
    \centering
    \includegraphics[width=\textwidth]{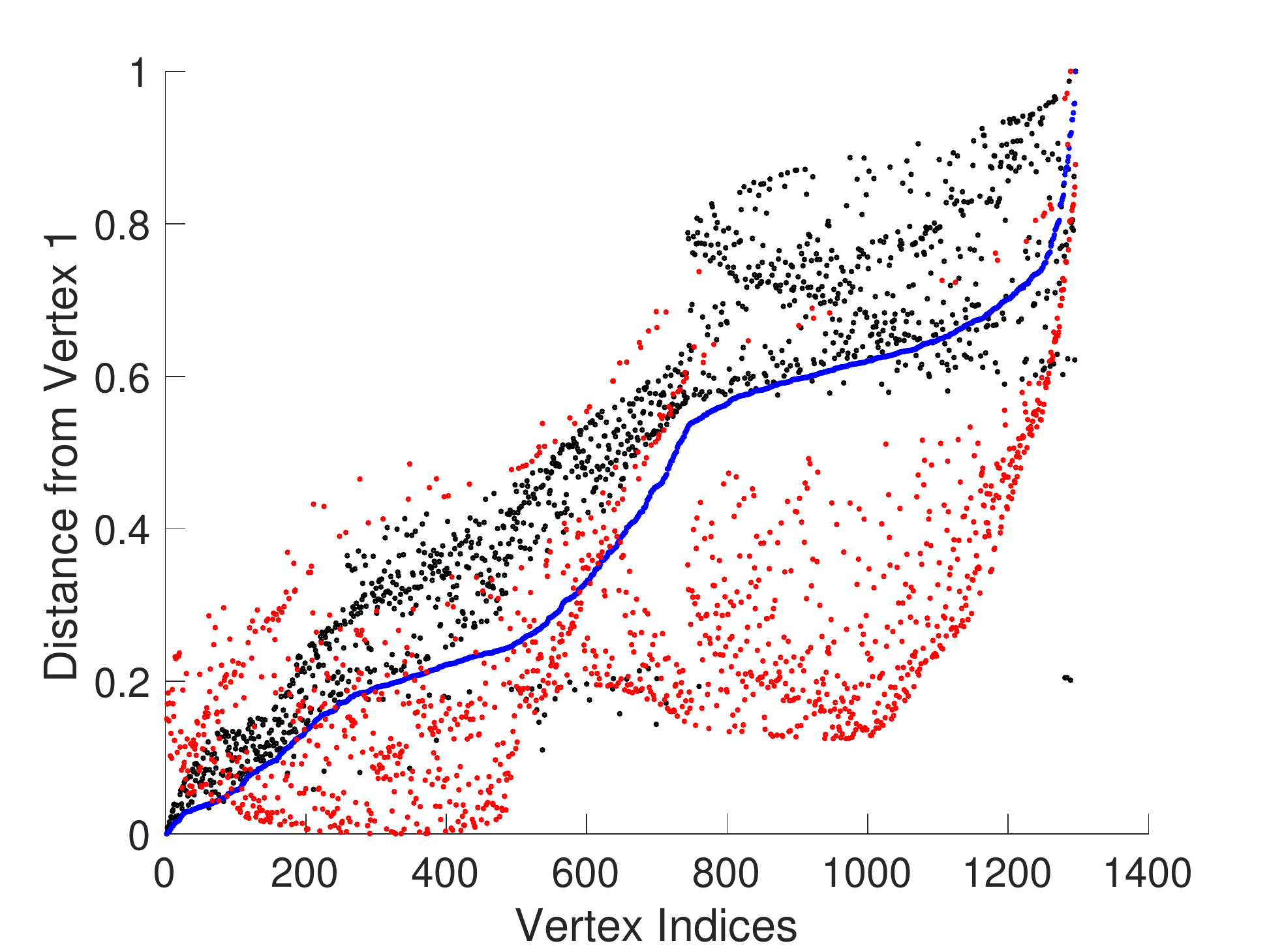}
    \caption{$H$ shaped domain}
  \end{subfigure}
  \\
  \begin{subfigure}[t]{0.3\textwidth}
    \centering
    \includegraphics[width=\textwidth]{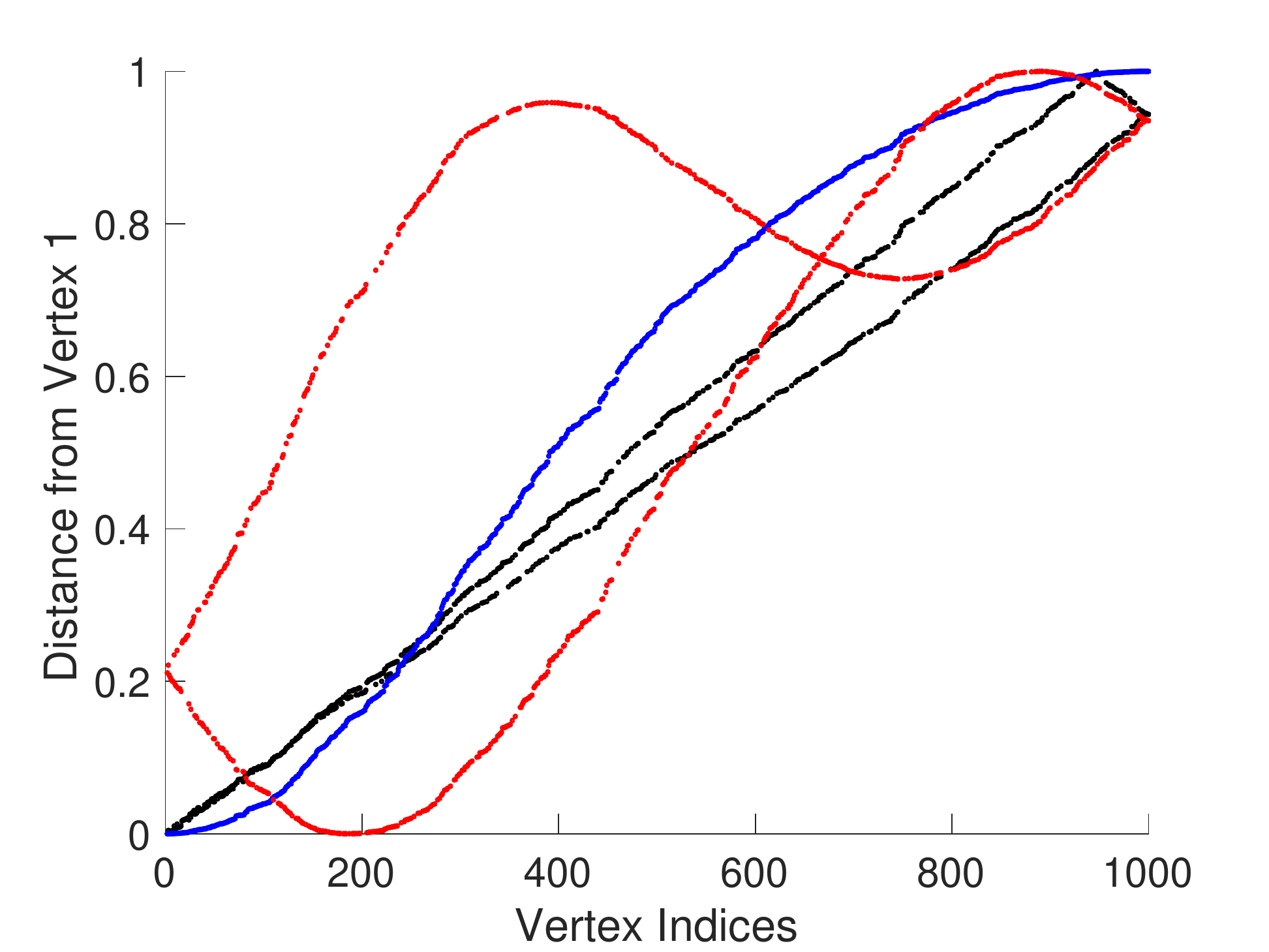}
    \caption{Circle}
  \end{subfigure}\quad
  \begin{subfigure}[t]{0.3\textwidth}
    \centering
    \includegraphics[width=\textwidth]{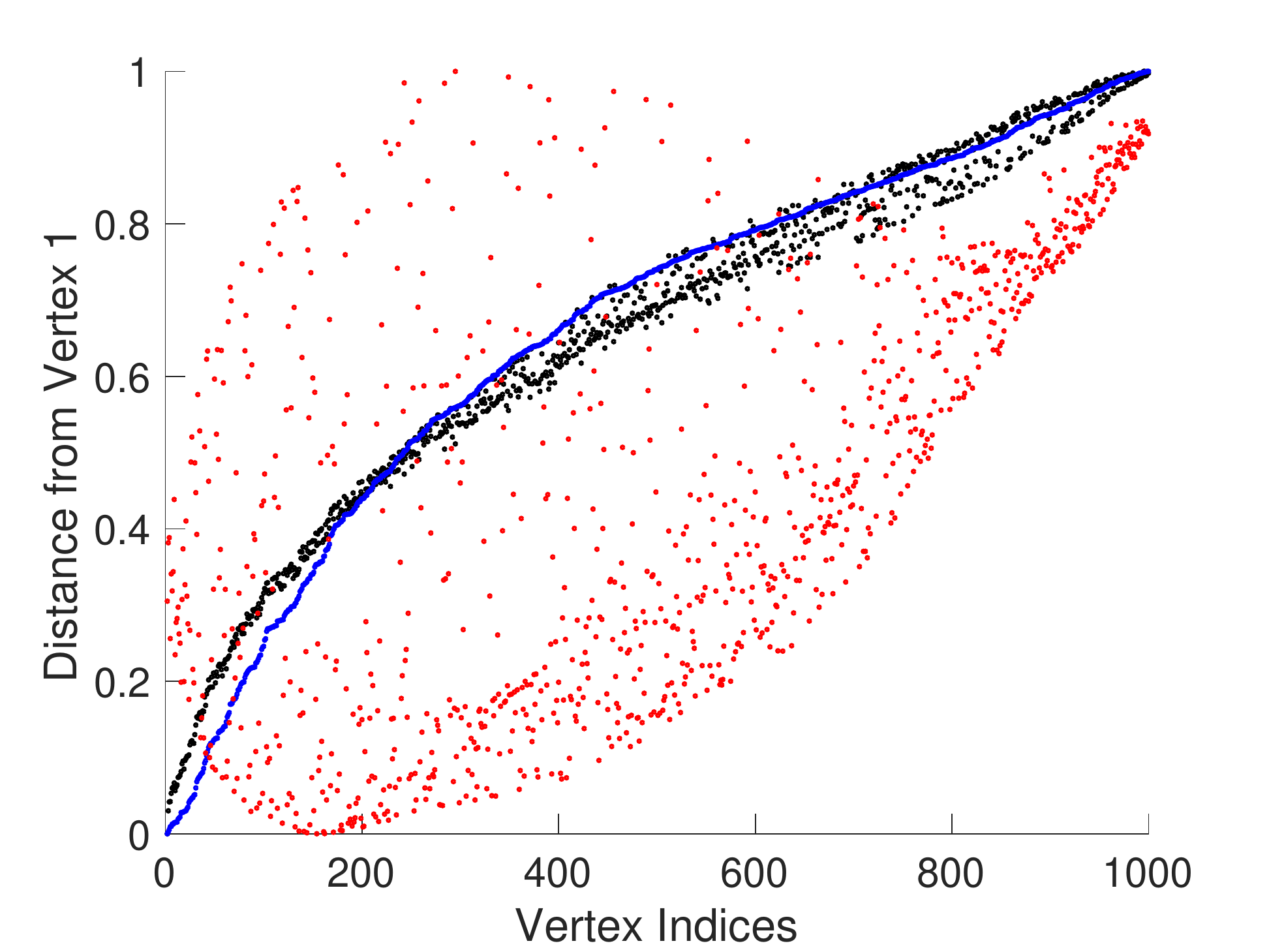}
    \caption{Sphere}
  \end{subfigure}\quad
  \begin{subfigure}[t]{0.3\textwidth}
    \centering
    \includegraphics[width=\textwidth]{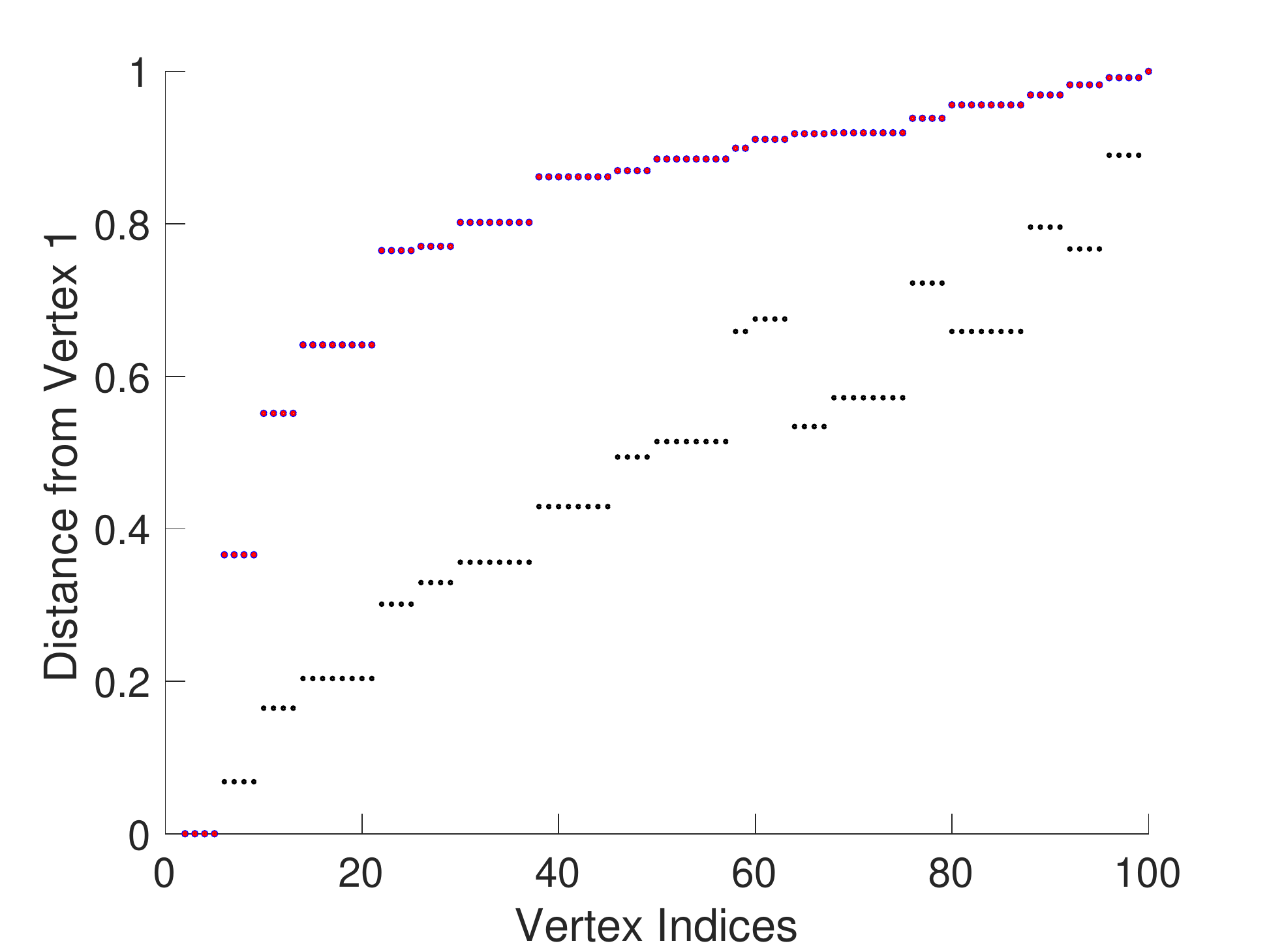}
    \caption{Square Lattice}
  \end{subfigure}  
  \caption{Normalized distance plots comparing the scaled distances from one node in a geometric graph computed using Euclidean distances (Black), $d^{\sfrac12}$ (Blue) and $d^1$ (Red). The geometric graphs from top left to bottom right are 
    {\bf (a)} A random point cloud on a flat torus 
    {\bf (b)} A random point cloud on a flat torus with a hole, 
    {\bf (c)} A random point cloud on an $H$ shaped domain, 
    {\bf (d)} A random point cloud on the circle, 
    {\bf (e)} A random point cloud on a sphere, 
    {\bf (f)} A square lattice on the flat torus.  
  Note, in all subplots, we have ordered the vertices from closest to farthest from a reference node given by the first vertex generated relative to the $d^{\sfrac12}$ metric.} 
  \label{fig:geodistance}
\end{figure}

\subsection{Real-world example: the New York City taxi network}%
\label{sec:taxi}%

Consider the movement over time of a New York City taxi, which we interpret as a Markov chain where the states are neighborhoods and $P_{i,j}$ is the probability that a trip begun in neighborhood $i$ ends in neighborhood $j$.
Using publicly available data from the New York City Taxi and Limousine Commission,\footnote{Accessed at~\url{https://www1.nyc.gov/site/tlc/about/tlc-trip-record-data.page} in April 2020.} 
we computed an adjacency matrix where $A_{i,j}$ is the number of Yellow Taxi trips in January 2019 that started at $i$ and ended at $j$, where $i$ and $j$ are chosen from 262 neighborhoods\footnote{Two additional neighborhoods are marked ``unknown'' and appear to designate out-of-city or out-of-state endpoints. We excluded these from our analysis.} spread across the city's five boroughs (Manhattan, Staten Island, Queens, Brooklyn, and the Bronx). We also included trips to and from Newark Liberty International Airport (EWR) in New Jersey. We restricted our analysis to the 250-neighborhood strongly connected component.

The data is dominated by degree, as shown in~\cref{nyc_heatmap}, with the busier Manhattan neighborhoods having tens of thousands of trips, and the Staten Island neighborhoods having median out-degree of 4. Traffic is also organized by borough, although the distinctions between the spatially adjacent Brooklyn, Queens, and Bronx boroughs are perhaps less apparent, and they might be properly considered as peripheral to the Manhattan core. Staten Island is notable for its remoteness, which is reflected in the sparsity of $A$ in that block.

In~\cref{nyc_heatmap}, we compare $A$ with $d^{\sfrac12}$ and $d^1$. While $d^{\sfrac12}$ highlights Manhattan in a manner similar to $A$, it does not distinguish much between Queens, Brooklyn, and the Bronx, showing them instead as a single interconnected group. In contrast, Staten Island is very clearly highlighted as its own, close group, which is reasonable given the geographic proximity of these neighborhoods and the fact that a disproportionately large number of trips involving Staten Island both started and ended there. Although the purpose of this example is not to provide an optimal clustering of the data, 
we note that Staten Island does represent a difficult cluster to detect, and arguably is not even a cluster, since there are only eight interior edges (counting multiplicity but excluding self-edges) and 309 incoming or outgoing edges, all of which is hidden in over 1 million edges (again, counting multiplicity).

Note that the fact that Staten Island is highlighted by $d^{\sfrac12}$ is not simply because of degree scaling, as a heat map of $P$ does not highlight Staten Island as a block. The true explanation seems to involve two factors: (1) Staten Island has eight non-diagonal in-edges 13 neighborhoods,\footnote{Staten Island has 20 neighborhoods, but 7 have 0 out-degree and are thus excluded from the strongly connected component.} and the median out-degree is 4. Thus, a taxi that does enter Staten Island has a relatively large likelihood of visiting another Staten Island location next, relative to taxis starting at other neighborhoods. (2) The average frequency of visiting Staten Island at all is so low that the pattern of visiting is almost memoryless, with taxis leaving Staten Island having plenty of time to mix in other areas before visiting Staten Island again, so that the probability of leaving Staten Island and then reaching another Staten Island location before returning to the first one is about $\frac12$, despite the low degree of Staten Island neighborhoods. In contrast, Staten Island is far from other locations, especially Manhattan, since by~\cref{d2sym}, mutually high hitting probabilities are required for closeness, but the probability of starting in a Manhattan neighborhood and reaching Staten Island before returning is very low. 

The distance $d^1$ places the Manhattan nodes close to most other nodes, especially each other, while the Staten Island nodes are far from everything, especially each other. Since $d^1_{i,j} = -\log(\phi_i)  - \log(Q_{i,j})$, this distance is small only when (1) $\phi_i$ is large and (2) $Q_{i,j}$ is far from zero. Thus, the Staten Island nodes, which have small values of $\phi$, cannot be close to anything, and the Manhattan nodes, which have the largest values of $\phi$, can be close to other nodes, depending on $Q_{i,j}$. Empirically, $Q_{i,j}$ is usually not very small, with 77\% of the entries in $Q$ being at least $0.1$, which explains Manhattan's overall closeness to other nodes. The fact that the Manhattan nodes are closer to each other than to other nodes is accounted for by the fact that $Q_{i,j}$ for $i$ in Manhattan is generally larger if $j$ is also in Manhattan, which might be expected. (The medians differ by a factor of $5.4$.) A similar observation explains why the Staten Island nodes are considered farther from each other than they are from nodes in the other boroughs.

Finally, we used $d^{\sfrac12}$ to perform PCA, with the first two principal components (PCs) visualized in~\cref{nyc_pca}. These two PCs explained 64\% and 34\% of the variation, respectively, with the first PC being closely related to out-degree (Pearson correlation with $\log k_{\mathrm{out}}$ is $.96$) and the second PC being well-correlated (Pearson correlation $.9978$) with the column means of $d^{\sfrac12}$. So over 98\% of the variance is explained by these two PCs. Interestingly, both the highest- and lowest-degree nodes were on average far from other nodes. Recalling that $d^{\sfrac12}_{i,j}$ is the negative log of the geometric mean of $Q_{i,j}$ and $Q_{j,i}$ (see~\cref{d2sym}), closeness requires that both of these factors be high. If $i$ is a high-degree core node, then $Q_{i,j}$ is small for most $j$. In contrast, a mid-degree peripheral node in Queens, the Bronx, or Brooklyn, enjoys reasonable values of $Q_{i,j}$ for other peripheral nodes $j$, since once a taxi enters the Manhattan core, it is likely to visit a significant portion of the other nodes before returning to $i$. Finally, if a node's degree is too small, the probability of a taxi reaching it at all is too small for the hitting probabilities to be high. For comparison, performing PCA directly on $A$ gives a similar first PC, with a different second PC that explains about half as much variance as the second PC of $d^{\sfrac12}$. The second PC is nearly constant, except on Manhattan, where it correlates with the East-West coordinate.

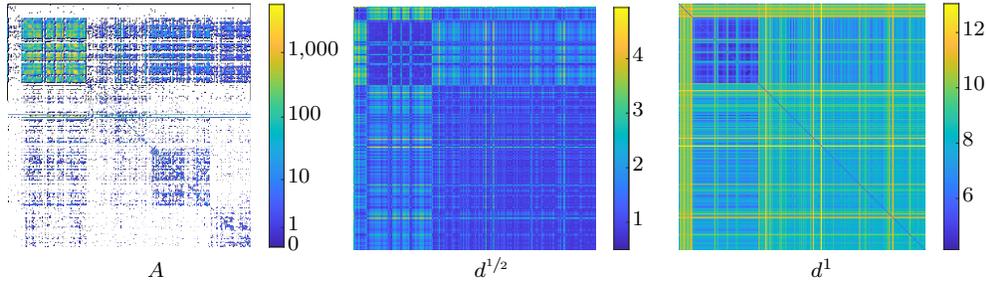
\begin{figure}[t] 
  \centering
  \footnotesize
  \begin{tabular}{@{\hskip 0in}c@{\hskip 0.18in}c@{\hskip 0.18in}c@{\hskip 0in}}
    \def\svgwidth{.31\linewidth}
    \input{nyc_adjacencyinput.tex} &
    \def\svgwidth{.29\linewidth}
    \input{nyc_distanceinput.tex} &
    \def\svgwidth{.30\linewidth}
    \input{nyc_d1input.tex}\\
    $A$ & $d^{\sfrac12}$ & $d^1$
  \end{tabular}

  \caption{An example based on New York City taxi transit data, where nodes are 250 neighborhoods and $A_{i,j}$ is the number of trips from $i$ to $j$. (Left) Heatmap of $A$ with the nodes sorted by borough in this order: EWR airport (one node), Staten Island, Manhattan, Queens, Brooklyn, and the Bronx. The taxi traffic is dominated by the Manhattan block in the upper left, with Queens, Brooklyn, and the Bronx forming three blocks further down and to the right. (Middle) A similarly arranged heatmap of $d^{\sfrac12}$. The Manhattan neighborhoods are close together, and the smaller upper left block corresponding to Staten Island is distinguished as a coherent submodule, despite having only  8 interior edges. Queens, Brooklyn, and the Bronx form a large block in the lower right.  See~\cref{sec:taxi} for an explanation of these differences.  (Right) A heatmap of $d^{1}$.  We observe that in this normalization Staten Island is quite far from everything, including itself. Manhattan is relatively close to almost everything, especially itself. This is exactly what we should expect because $d^1_{i,j}$ is small only when both $\phi_i$ and the $i\rightarrow j$ hitting probability are large. [In the right two heatmaps, the diagonal is set to a non-zero value to improve contrast.]}%
  \label{nyc_heatmap}
\end{figure}

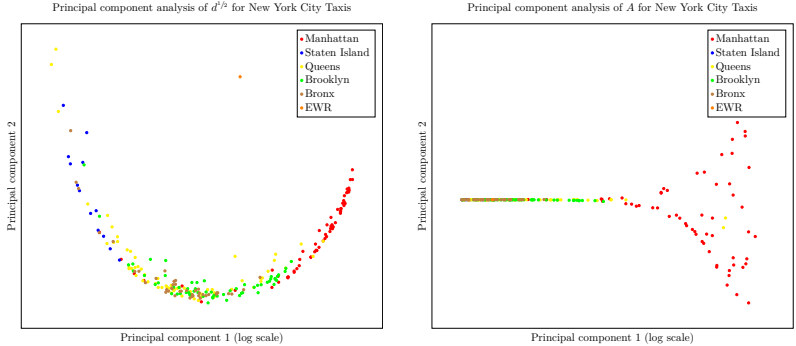
\begin{figure}[t]
  \centering
  \begin{tabular}{cc}
    \begin{tikzpicture}[scale=0.42]
      \begin{axis}[
	  cycle list={
	    {red,mark=*},
	    {blue,mark=*},
	    {yellow,mark=*},
	    {green,mark=*},
	    {brown,mark=*},
	    {orange,mark=*}
	  },
	  legend entries={Manhattan, Staten Island, Queens, Brooklyn, Bronx,EWR},
	  legend cell align=left,
	  only marks,
	  mark size=1pt,
	  width=\textwidth,
	  ticks=none,
	  xlabel={Principal component 1 (log scale)},
	  ylabel={Principal component 2},
	  title={Principal component analysis of $d^{\sfrac12}$ for New York City Taxis},
	]
	\addplot table [col sep=comma] {pca3.dat};
	\addplot table [col sep=comma] {pca2.dat};
	\addplot table [col sep=comma] {pca4.dat};
	\addplot table [col sep=comma] {pca5.dat};
	\addplot table [col sep=comma] {pca6.dat};
	\addplot table [col sep=comma] {pca1.dat};
      \end{axis}
    \end{tikzpicture}
    &
    \begin{tikzpicture}[scale=0.42]
      \begin{axis}[
	  cycle list={
	    {red,mark=*},
	    {blue,mark=*},
	    {yellow,mark=*},
	    {green,mark=*},
	    {brown,mark=*},
	    {orange,mark=*}
	  },
	  legend entries={Manhattan, Staten Island, Queens, Brooklyn, Bronx,EWR},
	  legend cell align=left,
	  only marks,
	  mark size=1pt,
	  width=\textwidth,
	  ticks=none,
	  xlabel={Principal component 1 (log scale)},
	  ylabel={Principal component 2},
	  title={Principal component analysis of $A$ for New York City Taxis},
	]
	\addplot table [col sep=comma] {pcaA3.dat};
	\addplot table [col sep=comma] {pcaA2.dat};
	\addplot table [col sep=comma] {pcaA4.dat};
	\addplot table [col sep=comma] {pcaA5.dat};
	\addplot table [col sep=comma] {pcaA6.dat};
	\addplot table [col sep=comma] {pcaA1.dat};
      \end{axis}
    \end{tikzpicture}
  \end{tabular}
  \caption{(Left) These two PCs explain $98.3\%$ of the variance. The first PC has a $.96$ correlation coefficient with $\log k_{\mathrm{out}}$, and the second PC as a correlation coefficient of $.9978$ with the column means of $d^{\sfrac12}$, which we interpret as the average distance to other nodes. Notably, the highest-degree nodes also have high average distance to other nodes. This is also true of the lowest-degree nodes, while the mid-degree nodes in Queens, Brooklyn, and the Bronx are closer to other nodes on average. We interpret this by noting that, while high-degree nodes are common endpoints for trips, (so $Q_{i,j}$ might be high when $j$ is a high-degree node), they have a lot of self-loops, and the taxis that leave them tend to return relatively quickly (so $Q_{j,i}$ is low for most $i$). Using~\cref{d2sym}, we see that $d^{\sfrac12}_{i,j}$ will then not be very small for high-degree $i$. The mid-degree nodes, in contrast, send a lot of taxis into the Manhattan core, which are likely to mix through the city for a long time before returning (so $Q_{i,j}$ is not very small for almost all destinations $j$). (Right) PCA on $A$ gives a similar first PC. The second PC is nearly constant except on Manhattan, where it is correlated with the East-West coordinate (Pearson .42, p=.0004). The second PC explains about half as much variance for $A$ as for $d^{\sfrac12}$.}%
  \label{nyc_pca}
\end{figure}

\section{Conclusion}\label{s:Disc}
Given a probability transition matrix for an ergodic, finite-state, time-homogeneous Markov chain, we have constructed a family of (possibly pseudo-)metrics on the state space, which we refer to as hitting probability distances. Alternatively, this construction gives a metric on the nodes of a strongly connected, directed graph. In the cases where we do not obtain a proper metric, the degeneracies give global structural information, and we can quotient them away. Our metrics can be computed in $O(n^3)$ time and $O(n^2)$ space, in one example scaling up $10,000$ nodes and $\approx 38M$ edges on a desktop computer. Our metric captures different information compared to other directed graph metrics and captures multiscale structure in the taxi example. We have considered the utility of this metric for structure detection, dimension reduction, and visualization, finding in each case advantages of our method compared to existing techniques.

Some other possible applications include efficient nearest-neighbor search, new notions of graph curvature~\cite{van_Gennip_2014}, Cheeger inequalities, and provable optimality of weak recovery for dense, directed communities.  Additionally, in our experiments, we observed that several eigenvalues of the symmetrized adjacency matrix contained useful information about structure such as cycles, and it would be good to understand better which structures get encoded in leading eigenspaces. Empirically, it is important to know how commonly $d^{\sfrac12}$ is degenerate, and what useful structure is revealed in practice.
A natural theoretical question is consistency of the distances in the large graph limit as we approach a natural geometric object embedded in a standard Euclidean space~\cite{Osting_2017,singer2012vector,singer2017spectral,trillos2018variational,trillos2016consistency,Yuan2020}.

In terms of possible improvements to our method, an effective means of thresholding the symmetrized hitting probability matrix could improve scalability.  A natural question to pursue in a variety of settings would be the sparsification of $\Ahtb$ and its implications for spectral analysis and clustering applications.  In particular, the potentially sparse $P$ will map into a full (but symmetric) matrix $\Ahtb$.  
In large systems the $O(n^2)$ storage requirement may become a burden.  Hence, it is natural to ask: If we sparsify the $\Ahtb$ matrix to have a comparable number of edges to that of the original $P$, how much information can be stably preserved in the spectrum?  This will be a topic of future work on the hitting probability matrices we have constructed.

\bibliographystyle{siamplain}
\bibliography{refs}

\end{document}

%% file: clusteringstudyinput.tex
\begingroup%
  \makeatletter%
  \providecommand\color[2][]{%
    \errmessage{(Inkscape) Color is used for the text in Inkscape, but the package 'color.sty' is not loaded}%
    \renewcommand\color[2][]{}%
  }%
  \providecommand\transparent[1]{%
    \errmessage{(Inkscape) Transparency is used (non-zero) for the text in Inkscape, but the package 'transparent.sty' is not loaded}%
    \renewcommand\transparent[1]{}%
  }%
  \providecommand\rotatebox[2]{#2}%
  \newcommand*\fsize{\dimexpr\f@size pt\relax}%
  \newcommand*\lineheight[1]{\fontsize{\fsize}{#1\fsize}\selectfont}%
  \ifx\svgwidth\undefined%
    \setlength{\unitlength}{664.02909851bp}%
    \ifx\svgscale\undefined%
      \relax%
    \else%
      \setlength{\unitlength}{\unitlength * \real{\svgscale}}%
    \fi%
  \else%
    \setlength{\unitlength}{\svgwidth}%
  \fi%
  \global\let\svgwidth\undefined%
  \global\let\svgscale\undefined%
  \makeatother%
  \begin{picture}(1,0.98163869)%
    \lineheight{1}%
    \setlength\tabcolsep{0pt}%
    \put(0.24755336,0.96905273){\makebox(0,0)[t]{\lineheight{1.25}\smash{\begin{tabular}[t]{c}Accuracy\end{tabular}}}}%
    \put(0,0){\includegraphics[width=\unitlength,page=1]{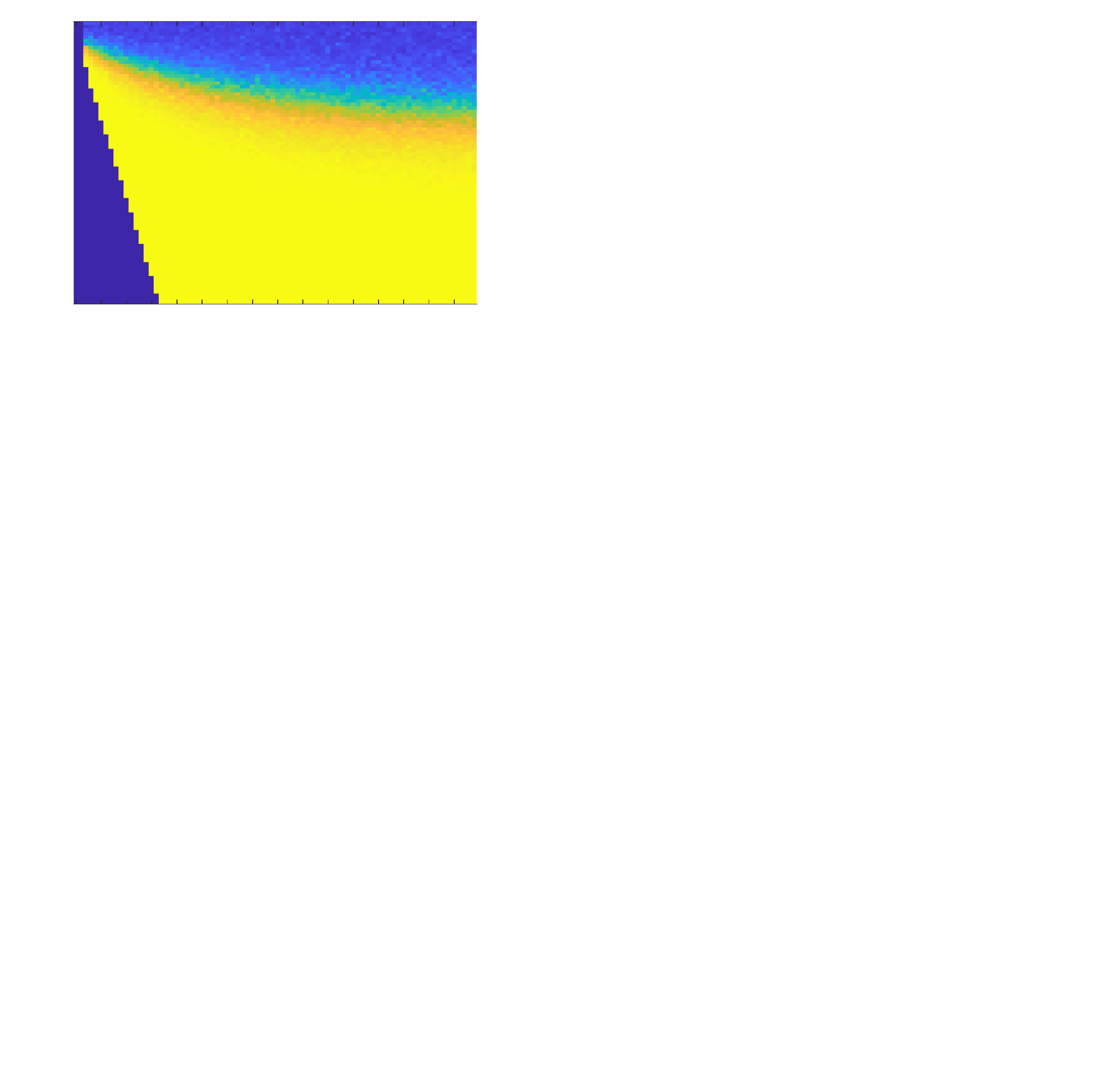}}%
    \put(0.06853938,0.6904317){\makebox(0,0)[t]{\lineheight{1.25}\smash{\begin{tabular}[t]{c}1\end{tabular}}}}%
    \put(0.1138593,0.6904317){\makebox(0,0)[t]{\lineheight{1.25}\smash{\begin{tabular}[t]{c}7\end{tabular}}}}%
    \put(0.15917923,0.6904317){\makebox(0,0)[t]{\lineheight{1.25}\smash{\begin{tabular}[t]{c}13\end{tabular}}}}%
    \put(0.20449915,0.6904317){\makebox(0,0)[t]{\lineheight{1.25}\smash{\begin{tabular}[t]{c}20\end{tabular}}}}%
    \put(0.24981908,0.6904317){\makebox(0,0)[t]{\lineheight{1.25}\smash{\begin{tabular}[t]{c}26\end{tabular}}}}%
    \put(0.29513901,0.6904317){\makebox(0,0)[t]{\lineheight{1.25}\smash{\begin{tabular}[t]{c}32\end{tabular}}}}%
    \put(0.34045893,0.6904317){\makebox(0,0)[t]{\lineheight{1.25}\smash{\begin{tabular}[t]{c}38\end{tabular}}}}%
    \put(0.38577886,0.6904317){\makebox(0,0)[t]{\lineheight{1.25}\smash{\begin{tabular}[t]{c}44\end{tabular}}}}%
    \put(0.24755325,0.66882113){\makebox(0,0)[t]{\lineheight{1.25}\smash{\begin{tabular}[t]{c}$100\rho$\end{tabular}}}}%
    \put(0,0){\includegraphics[width=\unitlength,page=2]{clusteringstudy.pdf}}%
    \put(0.06024953,0.95667805){\makebox(0,0)[rt]{\lineheight{1.25}\smash{\begin{tabular}[t]{r}1\end{tabular}}}}%
    \put(0.06024953,0.92491175){\makebox(0,0)[rt]{\lineheight{1.25}\smash{\begin{tabular}[t]{r}5\end{tabular}}}}%
    \put(0.06024953,0.89314544){\makebox(0,0)[rt]{\lineheight{1.25}\smash{\begin{tabular}[t]{r}9\end{tabular}}}}%
    \put(0.06024953,0.86137915){\makebox(0,0)[rt]{\lineheight{1.25}\smash{\begin{tabular}[t]{r}14\end{tabular}}}}%
    \put(0.06024953,0.82961284){\makebox(0,0)[rt]{\lineheight{1.25}\smash{\begin{tabular}[t]{r}18\end{tabular}}}}%
    \put(0.06024953,0.79784654){\makebox(0,0)[rt]{\lineheight{1.25}\smash{\begin{tabular}[t]{r}22\end{tabular}}}}%
    \put(0.06024953,0.76608023){\makebox(0,0)[rt]{\lineheight{1.25}\smash{\begin{tabular}[t]{r}26\end{tabular}}}}%
    \put(0.06024953,0.73431393){\makebox(0,0)[rt]{\lineheight{1.25}\smash{\begin{tabular}[t]{r}30\end{tabular}}}}%
    \put(0.02758596,0.83526729){\rotatebox{90}{\makebox(0,0)[t]{\lineheight{1.25}\smash{\begin{tabular}[b]{c}PCA + $k$-means on $A$ \\ $100 \Delta$ \end{tabular}}}}}%
    \put(0,0){\includegraphics[width=\unitlength,page=3]{clusteringstudy.pdf}}%
    \put(0.45827425,0.73068687){\makebox(0,0)[lt]{\lineheight{1.25}\smash{\begin{tabular}[t]{l}0.4\end{tabular}}}}%
    \put(0.45827425,0.76861679){\makebox(0,0)[lt]{\lineheight{1.25}\smash{\begin{tabular}[t]{l}0.5\end{tabular}}}}%
    \put(0.45827425,0.80654661){\makebox(0,0)[lt]{\lineheight{1.25}\smash{\begin{tabular}[t]{l}0.6\end{tabular}}}}%
    \put(0.45827425,0.84447653){\makebox(0,0)[lt]{\lineheight{1.25}\smash{\begin{tabular}[t]{l}0.7\end{tabular}}}}%
    \put(0.45827425,0.88240646){\makebox(0,0)[lt]{\lineheight{1.25}\smash{\begin{tabular}[t]{l}0.8\end{tabular}}}}%
    \put(0.45827425,0.92033638){\makebox(0,0)[lt]{\lineheight{1.25}\smash{\begin{tabular}[t]{l}0.9\end{tabular}}}}%
    \put(0.45827425,0.95826631){\makebox(0,0)[lt]{\lineheight{1.25}\smash{\begin{tabular}[t]{l}1\end{tabular}}}}%
    \put(0,0){\includegraphics[width=\unitlength,page=4]{clusteringstudy.pdf}}%
    \put(0.76033255,0.96905273){\makebox(0,0)[t]{\lineheight{1.25}\smash{\begin{tabular}[t]{c}p value\end{tabular}}}}%
    \put(0,0){\includegraphics[width=\unitlength,page=5]{clusteringstudy.pdf}}%
    \put(0.58131799,0.6904317){\makebox(0,0)[t]{\lineheight{1.25}\smash{\begin{tabular}[t]{c}1\end{tabular}}}}%
    \put(0.62663792,0.6904317){\makebox(0,0)[t]{\lineheight{1.25}\smash{\begin{tabular}[t]{c}7\end{tabular}}}}%
    \put(0.67195785,0.6904317){\makebox(0,0)[t]{\lineheight{1.25}\smash{\begin{tabular}[t]{c}13\end{tabular}}}}%
    \put(0.71727777,0.6904317){\makebox(0,0)[t]{\lineheight{1.25}\smash{\begin{tabular}[t]{c}20\end{tabular}}}}%
    \put(0.7625977,0.6904317){\makebox(0,0)[t]{\lineheight{1.25}\smash{\begin{tabular}[t]{c}26\end{tabular}}}}%
    \put(0.80791763,0.6904317){\makebox(0,0)[t]{\lineheight{1.25}\smash{\begin{tabular}[t]{c}32\end{tabular}}}}%
    \put(0.85323755,0.6904317){\makebox(0,0)[t]{\lineheight{1.25}\smash{\begin{tabular}[t]{c}38\end{tabular}}}}%
    \put(0.89855748,0.6904317){\makebox(0,0)[t]{\lineheight{1.25}\smash{\begin{tabular}[t]{c}44\end{tabular}}}}%
    \put(0.76033198,0.66882113){\makebox(0,0)[t]{\lineheight{1.25}\smash{\begin{tabular}[t]{c}$100\rho$\end{tabular}}}}%
    \put(0,0){\includegraphics[width=\unitlength,page=6]{clusteringstudy.pdf}}%
    \put(0.57302826,0.95667805){\makebox(0,0)[rt]{\lineheight{1.25}\smash{\begin{tabular}[t]{r}1\end{tabular}}}}%
    \put(0.57302826,0.92491175){\makebox(0,0)[rt]{\lineheight{1.25}\smash{\begin{tabular}[t]{r}5\end{tabular}}}}%
    \put(0.57302826,0.89314544){\makebox(0,0)[rt]{\lineheight{1.25}\smash{\begin{tabular}[t]{r}9\end{tabular}}}}%
    \put(0.57302826,0.86137915){\makebox(0,0)[rt]{\lineheight{1.25}\smash{\begin{tabular}[t]{r}14\end{tabular}}}}%
    \put(0.57302826,0.82961284){\makebox(0,0)[rt]{\lineheight{1.25}\smash{\begin{tabular}[t]{r}18\end{tabular}}}}%
    \put(0.57302826,0.79784654){\makebox(0,0)[rt]{\lineheight{1.25}\smash{\begin{tabular}[t]{r}22\end{tabular}}}}%
    \put(0.57302826,0.76608023){\makebox(0,0)[rt]{\lineheight{1.25}\smash{\begin{tabular}[t]{r}26\end{tabular}}}}%
    \put(0.57302826,0.73431393){\makebox(0,0)[rt]{\lineheight{1.25}\smash{\begin{tabular}[t]{r}30\end{tabular}}}}%
    \put(0.54036468,0.8352673){\rotatebox{90}{\makebox(0,0)[t]{\lineheight{1.25}\smash{\begin{tabular}[t]{c}$100\Delta$\end{tabular}}}}}%
    \put(0,0){\includegraphics[width=\unitlength,page=7]{clusteringstudy.pdf}}%
    \put(0.97105298,0.77447015){\makebox(0,0)[lt]{\lineheight{1.25}\smash{\begin{tabular}[t]{l}.1\end{tabular}}}}%
    \put(0.97105298,0.84510508){\makebox(0,0)[lt]{\lineheight{1.25}\smash{\begin{tabular}[t]{l}.01\end{tabular}}}}%
    \put(0.97105298,0.9157399){\makebox(0,0)[lt]{\lineheight{1.25}\smash{\begin{tabular}[t]{l}.001\end{tabular}}}}%
    \put(0,0){\includegraphics[width=\unitlength,page=8]{clusteringstudy.pdf}}%
    \put(0.06853938,0.35723847){\makebox(0,0)[t]{\lineheight{1.25}\smash{\begin{tabular}[t]{c}1\end{tabular}}}}%
    \put(0.1138593,0.35723847){\makebox(0,0)[t]{\lineheight{1.25}\smash{\begin{tabular}[t]{c}7\end{tabular}}}}%
    \put(0.15917923,0.35723847){\makebox(0,0)[t]{\lineheight{1.25}\smash{\begin{tabular}[t]{c}13\end{tabular}}}}%
    \put(0.20449915,0.35723847){\makebox(0,0)[t]{\lineheight{1.25}\smash{\begin{tabular}[t]{c}20\end{tabular}}}}%
    \put(0.24981908,0.35723847){\makebox(0,0)[t]{\lineheight{1.25}\smash{\begin{tabular}[t]{c}26\end{tabular}}}}%
    \put(0.29513901,0.35723847){\makebox(0,0)[t]{\lineheight{1.25}\smash{\begin{tabular}[t]{c}32\end{tabular}}}}%
    \put(0.34045893,0.35723847){\makebox(0,0)[t]{\lineheight{1.25}\smash{\begin{tabular}[t]{c}38\end{tabular}}}}%
    \put(0.38577886,0.35723847){\makebox(0,0)[t]{\lineheight{1.25}\smash{\begin{tabular}[t]{c}44\end{tabular}}}}%
    \put(0.24755325,0.3356279){\makebox(0,0)[t]{\lineheight{1.25}\smash{\begin{tabular}[t]{c}$100\rho$\end{tabular}}}}%
    \put(0,0){\includegraphics[width=\unitlength,page=9]{clusteringstudy.pdf}}%
    \put(0.06024953,0.62348483){\makebox(0,0)[rt]{\lineheight{1.25}\smash{\begin{tabular}[t]{r}1\end{tabular}}}}%
    \put(0.06024953,0.59171852){\makebox(0,0)[rt]{\lineheight{1.25}\smash{\begin{tabular}[t]{r}5\end{tabular}}}}%
    \put(0.06024953,0.55995222){\makebox(0,0)[rt]{\lineheight{1.25}\smash{\begin{tabular}[t]{r}9\end{tabular}}}}%
    \put(0.06024953,0.52818592){\makebox(0,0)[rt]{\lineheight{1.25}\smash{\begin{tabular}[t]{r}14\end{tabular}}}}%
    \put(0.06024953,0.49641961){\makebox(0,0)[rt]{\lineheight{1.25}\smash{\begin{tabular}[t]{r}18\end{tabular}}}}%
    \put(0.06024953,0.46465331){\makebox(0,0)[rt]{\lineheight{1.25}\smash{\begin{tabular}[t]{r}22\end{tabular}}}}%
    \put(0.06024953,0.432887){\makebox(0,0)[rt]{\lineheight{1.25}\smash{\begin{tabular}[t]{r}26\end{tabular}}}}%
    \put(0.06024953,0.4011207){\makebox(0,0)[rt]{\lineheight{1.25}\smash{\begin{tabular}[t]{r}30\end{tabular}}}}%
    \put(0.02758595,0.50207407){\rotatebox{90}{\makebox(0,0)[t]{\lineheight{1.25}\smash{\begin{tabular}[b]{c}PCA + $k$-means on $d^{\sfrac12}$ \\ $100 \Delta$ \end{tabular}}}}}%
    \put(0,0){\includegraphics[width=\unitlength,page=10]{clusteringstudy.pdf}}%
    \put(0.45827425,0.39749353){\makebox(0,0)[lt]{\lineheight{1.25}\smash{\begin{tabular}[t]{l}0.4\end{tabular}}}}%
    \put(0.45827425,0.43542345){\makebox(0,0)[lt]{\lineheight{1.25}\smash{\begin{tabular}[t]{l}0.5\end{tabular}}}}%
    \put(0.45827425,0.47335338){\makebox(0,0)[lt]{\lineheight{1.25}\smash{\begin{tabular}[t]{l}0.6\end{tabular}}}}%
    \put(0.45827425,0.51128331){\makebox(0,0)[lt]{\lineheight{1.25}\smash{\begin{tabular}[t]{l}0.7\end{tabular}}}}%
    \put(0.45827425,0.54921323){\makebox(0,0)[lt]{\lineheight{1.25}\smash{\begin{tabular}[t]{l}0.8\end{tabular}}}}%
    \put(0.45827425,0.58714316){\makebox(0,0)[lt]{\lineheight{1.25}\smash{\begin{tabular}[t]{l}0.9\end{tabular}}}}%
    \put(0.45827425,0.62507309){\makebox(0,0)[lt]{\lineheight{1.25}\smash{\begin{tabular}[t]{l}1\end{tabular}}}}%
    \put(0,0){\includegraphics[width=\unitlength,page=11]{clusteringstudy.pdf}}%
    \put(0.58131799,0.35723847){\makebox(0,0)[t]{\lineheight{1.25}\smash{\begin{tabular}[t]{c}1\end{tabular}}}}%
    \put(0.62663792,0.35723847){\makebox(0,0)[t]{\lineheight{1.25}\smash{\begin{tabular}[t]{c}7\end{tabular}}}}%
    \put(0.67195785,0.35723847){\makebox(0,0)[t]{\lineheight{1.25}\smash{\begin{tabular}[t]{c}13\end{tabular}}}}%
    \put(0.71727777,0.35723847){\makebox(0,0)[t]{\lineheight{1.25}\smash{\begin{tabular}[t]{c}20\end{tabular}}}}%
    \put(0.7625977,0.35723847){\makebox(0,0)[t]{\lineheight{1.25}\smash{\begin{tabular}[t]{c}26\end{tabular}}}}%
    \put(0.80791763,0.35723847){\makebox(0,0)[t]{\lineheight{1.25}\smash{\begin{tabular}[t]{c}32\end{tabular}}}}%
    \put(0.85323755,0.35723847){\makebox(0,0)[t]{\lineheight{1.25}\smash{\begin{tabular}[t]{c}38\end{tabular}}}}%
    \put(0.89855748,0.35723847){\makebox(0,0)[t]{\lineheight{1.25}\smash{\begin{tabular}[t]{c}44\end{tabular}}}}%
    \put(0.76033198,0.3356279){\makebox(0,0)[t]{\lineheight{1.25}\smash{\begin{tabular}[t]{c}$100\rho$\end{tabular}}}}%
    \put(0,0){\includegraphics[width=\unitlength,page=12]{clusteringstudy.pdf}}%
    \put(0.57302826,0.62348483){\makebox(0,0)[rt]{\lineheight{1.25}\smash{\begin{tabular}[t]{r}1\end{tabular}}}}%
    \put(0.57302826,0.59171852){\makebox(0,0)[rt]{\lineheight{1.25}\smash{\begin{tabular}[t]{r}5\end{tabular}}}}%
    \put(0.57302826,0.55995222){\makebox(0,0)[rt]{\lineheight{1.25}\smash{\begin{tabular}[t]{r}9\end{tabular}}}}%
    \put(0.57302826,0.52818592){\makebox(0,0)[rt]{\lineheight{1.25}\smash{\begin{tabular}[t]{r}14\end{tabular}}}}%
    \put(0.57302826,0.49641961){\makebox(0,0)[rt]{\lineheight{1.25}\smash{\begin{tabular}[t]{r}18\end{tabular}}}}%
    \put(0.57302826,0.46465331){\makebox(0,0)[rt]{\lineheight{1.25}\smash{\begin{tabular}[t]{r}22\end{tabular}}}}%
    \put(0.57302826,0.432887){\makebox(0,0)[rt]{\lineheight{1.25}\smash{\begin{tabular}[t]{r}26\end{tabular}}}}%
    \put(0.57302826,0.4011207){\makebox(0,0)[rt]{\lineheight{1.25}\smash{\begin{tabular}[t]{r}30\end{tabular}}}}%
    \put(0.54036468,0.50207407){\rotatebox{90}{\makebox(0,0)[t]{\lineheight{1.25}\smash{\begin{tabular}[t]{c}$100\Delta$\end{tabular}}}}}%
    \put(0,0){\includegraphics[width=\unitlength,page=13]{clusteringstudy.pdf}}%
    \put(0.97105298,0.4414774){\makebox(0,0)[lt]{\lineheight{1.25}\smash{\begin{tabular}[t]{l}.1\end{tabular}}}}%
    \put(0.97105298,0.51203531){\makebox(0,0)[lt]{\lineheight{1.25}\smash{\begin{tabular}[t]{l}.01\end{tabular}}}}%
    \put(0.97105298,0.58259309){\makebox(0,0)[lt]{\lineheight{1.25}\smash{\begin{tabular}[t]{l}.001\end{tabular}}}}%
    \put(0,0){\includegraphics[width=\unitlength,page=14]{clusteringstudy.pdf}}%
    \put(0.06853938,0.02404525){\makebox(0,0)[t]{\lineheight{1.25}\smash{\begin{tabular}[t]{c}1\end{tabular}}}}%
    \put(0.1138593,0.02404525){\makebox(0,0)[t]{\lineheight{1.25}\smash{\begin{tabular}[t]{c}7\end{tabular}}}}%
    \put(0.15917923,0.02404525){\makebox(0,0)[t]{\lineheight{1.25}\smash{\begin{tabular}[t]{c}13\end{tabular}}}}%
    \put(0.20449915,0.02404525){\makebox(0,0)[t]{\lineheight{1.25}\smash{\begin{tabular}[t]{c}20\end{tabular}}}}%
    \put(0.24981908,0.02404525){\makebox(0,0)[t]{\lineheight{1.25}\smash{\begin{tabular}[t]{c}26\end{tabular}}}}%
    \put(0.29513901,0.02404525){\makebox(0,0)[t]{\lineheight{1.25}\smash{\begin{tabular}[t]{c}32\end{tabular}}}}%
    \put(0.34045893,0.02404525){\makebox(0,0)[t]{\lineheight{1.25}\smash{\begin{tabular}[t]{c}38\end{tabular}}}}%
    \put(0.38577886,0.02404525){\makebox(0,0)[t]{\lineheight{1.25}\smash{\begin{tabular}[t]{c}44\end{tabular}}}}%
    \put(0.24755325,0.00243467){\makebox(0,0)[t]{\lineheight{1.25}\smash{\begin{tabular}[t]{c}$100\rho$\end{tabular}}}}%
    \put(0,0){\includegraphics[width=\unitlength,page=15]{clusteringstudy.pdf}}%
    \put(0.06024953,0.2902916){\makebox(0,0)[rt]{\lineheight{1.25}\smash{\begin{tabular}[t]{r}1\end{tabular}}}}%
    \put(0.06024953,0.25852529){\makebox(0,0)[rt]{\lineheight{1.25}\smash{\begin{tabular}[t]{r}5\end{tabular}}}}%
    \put(0.06024953,0.22675899){\makebox(0,0)[rt]{\lineheight{1.25}\smash{\begin{tabular}[t]{r}9\end{tabular}}}}%
    \put(0.06024953,0.19499269){\makebox(0,0)[rt]{\lineheight{1.25}\smash{\begin{tabular}[t]{r}14\end{tabular}}}}%
    \put(0.06024953,0.16322638){\makebox(0,0)[rt]{\lineheight{1.25}\smash{\begin{tabular}[t]{r}18\end{tabular}}}}%
    \put(0.06024953,0.13146008){\makebox(0,0)[rt]{\lineheight{1.25}\smash{\begin{tabular}[t]{r}22\end{tabular}}}}%
    \put(0.06024953,0.09969377){\makebox(0,0)[rt]{\lineheight{1.25}\smash{\begin{tabular}[t]{r}26\end{tabular}}}}%
    \put(0.06024953,0.06792747){\makebox(0,0)[rt]{\lineheight{1.25}\smash{\begin{tabular}[t]{r}30\end{tabular}}}}%
    \put(0.02758595,0.16888084){\rotatebox{90}{\makebox(0,0)[t]{\lineheight{1.25}\smash{\begin{tabular}[b]{c}$k$-medoids on $d^{\sfrac12}$ \\ $100 \Delta$ \end{tabular}}}}}%
    \put(0,0){\includegraphics[width=\unitlength,page=16]{clusteringstudy.pdf}}%
    \put(0.45827425,0.0643003){\makebox(0,0)[lt]{\lineheight{1.25}\smash{\begin{tabular}[t]{l}0.4\end{tabular}}}}%
    \put(0.45827425,0.10223022){\makebox(0,0)[lt]{\lineheight{1.25}\smash{\begin{tabular}[t]{l}0.5\end{tabular}}}}%
    \put(0.45827425,0.14016015){\makebox(0,0)[lt]{\lineheight{1.25}\smash{\begin{tabular}[t]{l}0.6\end{tabular}}}}%
    \put(0.45827425,0.17809008){\makebox(0,0)[lt]{\lineheight{1.25}\smash{\begin{tabular}[t]{l}0.7\end{tabular}}}}%
    \put(0.45827425,0.21602){\makebox(0,0)[lt]{\lineheight{1.25}\smash{\begin{tabular}[t]{l}0.8\end{tabular}}}}%
    \put(0.45827425,0.25394993){\makebox(0,0)[lt]{\lineheight{1.25}\smash{\begin{tabular}[t]{l}0.9\end{tabular}}}}%
    \put(0.45827425,0.29187986){\makebox(0,0)[lt]{\lineheight{1.25}\smash{\begin{tabular}[t]{l}1\end{tabular}}}}%
    \put(0,0){\includegraphics[width=\unitlength,page=17]{clusteringstudy.pdf}}%
    \put(0.58131799,0.02404525){\makebox(0,0)[t]{\lineheight{1.25}\smash{\begin{tabular}[t]{c}1\end{tabular}}}}%
    \put(0.62663792,0.02404525){\makebox(0,0)[t]{\lineheight{1.25}\smash{\begin{tabular}[t]{c}7\end{tabular}}}}%
    \put(0.67195785,0.02404525){\makebox(0,0)[t]{\lineheight{1.25}\smash{\begin{tabular}[t]{c}13\end{tabular}}}}%
    \put(0.71727777,0.02404525){\makebox(0,0)[t]{\lineheight{1.25}\smash{\begin{tabular}[t]{c}20\end{tabular}}}}%
    \put(0.7625977,0.02404525){\makebox(0,0)[t]{\lineheight{1.25}\smash{\begin{tabular}[t]{c}26\end{tabular}}}}%
    \put(0.80791763,0.02404525){\makebox(0,0)[t]{\lineheight{1.25}\smash{\begin{tabular}[t]{c}32\end{tabular}}}}%
    \put(0.85323755,0.02404525){\makebox(0,0)[t]{\lineheight{1.25}\smash{\begin{tabular}[t]{c}38\end{tabular}}}}%
    \put(0.89855748,0.02404525){\makebox(0,0)[t]{\lineheight{1.25}\smash{\begin{tabular}[t]{c}44\end{tabular}}}}%
    \put(0.76033198,0.00243467){\makebox(0,0)[t]{\lineheight{1.25}\smash{\begin{tabular}[t]{c}$100\rho$\end{tabular}}}}%
    \put(0,0){\includegraphics[width=\unitlength,page=18]{clusteringstudy.pdf}}%
    \put(0.57302826,0.2902916){\makebox(0,0)[rt]{\lineheight{1.25}\smash{\begin{tabular}[t]{r}1\end{tabular}}}}%
    \put(0.57302826,0.25852529){\makebox(0,0)[rt]{\lineheight{1.25}\smash{\begin{tabular}[t]{r}5\end{tabular}}}}%
    \put(0.57302826,0.22675899){\makebox(0,0)[rt]{\lineheight{1.25}\smash{\begin{tabular}[t]{r}9\end{tabular}}}}%
    \put(0.57302826,0.19499269){\makebox(0,0)[rt]{\lineheight{1.25}\smash{\begin{tabular}[t]{r}14\end{tabular}}}}%
    \put(0.57302826,0.16322638){\makebox(0,0)[rt]{\lineheight{1.25}\smash{\begin{tabular}[t]{r}18\end{tabular}}}}%
    \put(0.57302826,0.13146008){\makebox(0,0)[rt]{\lineheight{1.25}\smash{\begin{tabular}[t]{r}22\end{tabular}}}}%
    \put(0.57302826,0.09969377){\makebox(0,0)[rt]{\lineheight{1.25}\smash{\begin{tabular}[t]{r}26\end{tabular}}}}%
    \put(0.57302826,0.06792747){\makebox(0,0)[rt]{\lineheight{1.25}\smash{\begin{tabular}[t]{r}30\end{tabular}}}}%
    \put(0.54036468,0.16888084){\rotatebox{90}{\makebox(0,0)[t]{\lineheight{1.25}\smash{\begin{tabular}[t]{c}$100\Delta$\end{tabular}}}}}%
    \put(0,0){\includegraphics[width=\unitlength,page=19]{clusteringstudy.pdf}}%
    \put(0.97105298,0.10820647){\makebox(0,0)[lt]{\lineheight{1.25}\smash{\begin{tabular}[t]{l}.1\end{tabular}}}}%
    \put(0.97105298,0.17879419){\makebox(0,0)[lt]{\lineheight{1.25}\smash{\begin{tabular}[t]{l}.01\end{tabular}}}}%
    \put(0.97105298,0.24938191){\makebox(0,0)[lt]{\lineheight{1.25}\smash{\begin{tabular}[t]{l}.001\end{tabular}}}}%
    \put(0,0){\includegraphics[width=\unitlength,page=20]{clusteringstudy.pdf}}%
  \end{picture}%
\endgroup%

%% file: clusteringstudy1input.tex
\begingroup%
  \makeatletter%
  \providecommand\color[2][]{%
    \errmessage{(Inkscape) Color is used for the text in Inkscape, but the package 'color.sty' is not loaded}%
    \renewcommand\color[2][]{}%
  }%
  \providecommand\transparent[1]{%
    \errmessage{(Inkscape) Transparency is used (non-zero) for the text in Inkscape, but the package 'transparent.sty' is not loaded}%
    \renewcommand\transparent[1]{}%
  }%
  \providecommand\rotatebox[2]{#2}%
  \newcommand*\fsize{\dimexpr\f@size pt\relax}%
  \newcommand*\lineheight[1]{\fontsize{\fsize}{#1\fsize}\selectfont}%
  \ifx\svgwidth\undefined%
    \setlength{\unitlength}{1004.12722778bp}%
    \ifx\svgscale\undefined%
      \relax%
    \else%
      \setlength{\unitlength}{\unitlength * \real{\svgscale}}%
    \fi%
  \else%
    \setlength{\unitlength}{\svgwidth}%
  \fi%
  \global\let\svgwidth\undefined%
  \global\let\svgscale\undefined%
  \makeatother%
  \begin{picture}(1,0.34335607)%
    \lineheight{1}%
    \setlength\tabcolsep{0pt}%
    \put(0,0){\includegraphics[width=\unitlength,page=1]{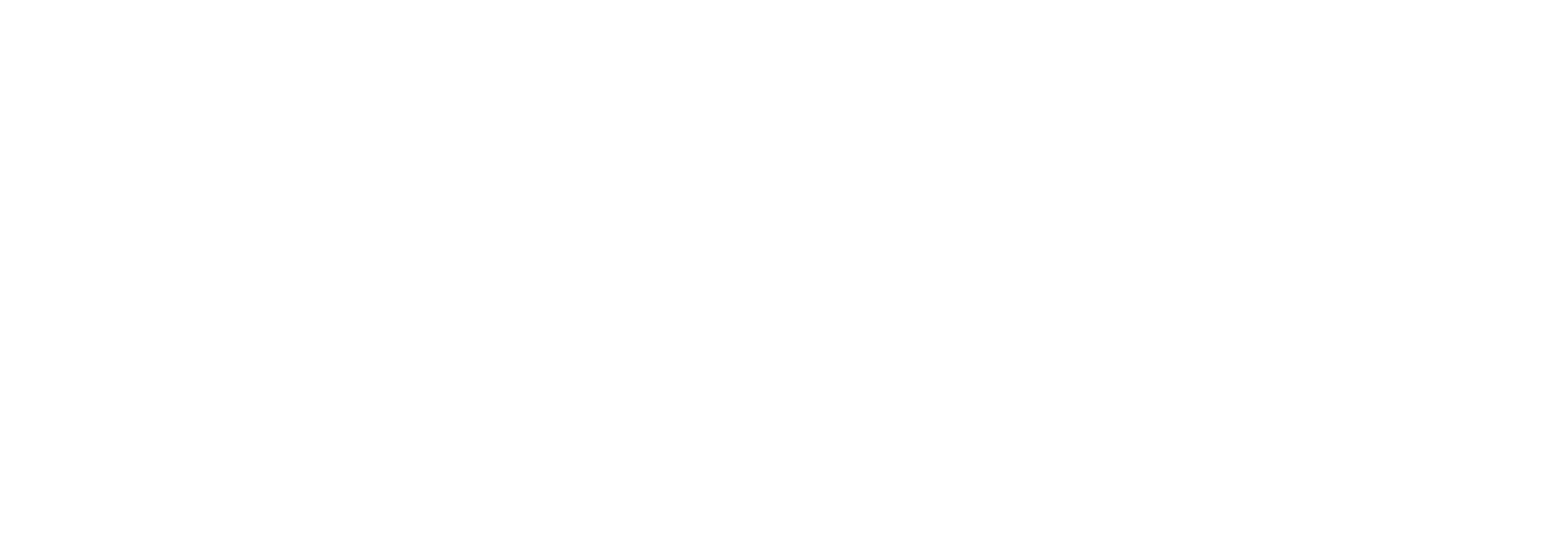}}%
    \put(0.16667518,0.33654628){\makebox(0,0)[t]{\lineheight{1.25}\smash{\begin{tabular}[t]{c}Best method\end{tabular}}}}%
    \put(0,0){\includegraphics[width=\unitlength,page=2]{clusteringstudy1.pdf}}%
    \put(0.04386779,0.01160838){\color[rgb]{0.14901961,0.14901961,0.14901961}\makebox(0,0)[t]{\lineheight{1.25}\smash{\begin{tabular}[t]{c}1\end{tabular}}}}%
    \put(0.07495822,0.01160838){\color[rgb]{0.14901961,0.14901961,0.14901961}\makebox(0,0)[t]{\lineheight{1.25}\smash{\begin{tabular}[t]{c}7\end{tabular}}}}%
    \put(0.10604865,0.01160838){\color[rgb]{0.14901961,0.14901961,0.14901961}\makebox(0,0)[t]{\lineheight{1.25}\smash{\begin{tabular}[t]{c}13\end{tabular}}}}%
    \put(0.13713901,0.01160838){\color[rgb]{0.14901961,0.14901961,0.14901961}\makebox(0,0)[t]{\lineheight{1.25}\smash{\begin{tabular}[t]{c}20\end{tabular}}}}%
    \put(0.16822945,0.01160838){\color[rgb]{0.14901961,0.14901961,0.14901961}\makebox(0,0)[t]{\lineheight{1.25}\smash{\begin{tabular}[t]{c}26\end{tabular}}}}%
    \put(0.19931988,0.01160838){\color[rgb]{0.14901961,0.14901961,0.14901961}\makebox(0,0)[t]{\lineheight{1.25}\smash{\begin{tabular}[t]{c}32\end{tabular}}}}%
    \put(0.23041031,0.01160838){\color[rgb]{0.14901961,0.14901961,0.14901961}\makebox(0,0)[t]{\lineheight{1.25}\smash{\begin{tabular}[t]{c}38\end{tabular}}}}%
    \put(0.26150074,0.01160838){\color[rgb]{0.14901961,0.14901961,0.14901961}\makebox(0,0)[t]{\lineheight{1.25}\smash{\begin{tabular}[t]{c}44\end{tabular}}}}%
    \put(0.16667512,-.01){\color[rgb]{0.14901961,0.14901961,0.14901961}\makebox(0,0)[t]{\lineheight{1.25}\smash{\begin{tabular}[t]{c}$100\rho$\end{tabular}}}}%
    \put(0,0){\includegraphics[width=\unitlength,page=3]{clusteringstudy1.pdf}}%
    \put(0.0383297,0.32306443){\color[rgb]{0.14901961,0.14901961,0.14901961}\makebox(0,0)[rt]{\lineheight{1.25}\smash{\begin{tabular}[t]{r}1.0\end{tabular}}}}%
    \put(0.0383297,0.27227405){\color[rgb]{0.14901961,0.14901961,0.14901961}\makebox(0,0)[rt]{\lineheight{1.25}\smash{\begin{tabular}[t]{r}2.7\end{tabular}}}}%
    \put(0.0383297,0.22148367){\color[rgb]{0.14901961,0.14901961,0.14901961}\makebox(0,0)[rt]{\lineheight{1.25}\smash{\begin{tabular}[t]{r}4.3\end{tabular}}}}%
    \put(0.0383297,0.1706933){\color[rgb]{0.14901961,0.14901961,0.14901961}\makebox(0,0)[rt]{\lineheight{1.25}\smash{\begin{tabular}[t]{r}6.0\end{tabular}}}}%
    \put(0.0383297,0.11990291){\color[rgb]{0.14901961,0.14901961,0.14901961}\makebox(0,0)[rt]{\lineheight{1.25}\smash{\begin{tabular}[t]{r}7.7\end{tabular}}}}%
    \put(0.0383297,0.06911253){\color[rgb]{0.14901961,0.14901961,0.14901961}\makebox(0,0)[rt]{\lineheight{1.25}\smash{\begin{tabular}[t]{r}9.4\end{tabular}}}}%
    \put(0.00,0.17973115){\color[rgb]{0.14901961,0.14901961,0.14901961}\rotatebox{90}{\makebox(0,0)[t]{\lineheight{1.25}\smash{\begin{tabular}[t]{c}$100\Delta$\end{tabular}}}}}%
    \put(0,0){\includegraphics[width=\unitlength,page=4]{clusteringstudy1.pdf}}%
    \put(0.50316173,0.33654628){\makebox(0,0)[t]{\lineheight{1.25}\smash{\begin{tabular}[t]{c}Accuracy difference\end{tabular}}}}%
    \put(0,0){\includegraphics[width=\unitlength,page=5]{clusteringstudy1.pdf}}%
    \put(0.38072279,0.01160838){\color[rgb]{0.14901961,0.14901961,0.14901961}\makebox(0,0)[t]{\lineheight{1.25}\smash{\begin{tabular}[t]{c}1\end{tabular}}}}%
    \put(0.41171988,0.01160838){\color[rgb]{0.14901961,0.14901961,0.14901961}\makebox(0,0)[t]{\lineheight{1.25}\smash{\begin{tabular}[t]{c}7\end{tabular}}}}%
    \put(0.44271695,0.01160838){\color[rgb]{0.14901961,0.14901961,0.14901961}\makebox(0,0)[t]{\lineheight{1.25}\smash{\begin{tabular}[t]{c}13\end{tabular}}}}%
    \put(0.47371402,0.01160838){\color[rgb]{0.14901961,0.14901961,0.14901961}\makebox(0,0)[t]{\lineheight{1.25}\smash{\begin{tabular}[t]{c}20\end{tabular}}}}%
    \put(0.50471109,0.01160838){\color[rgb]{0.14901961,0.14901961,0.14901961}\makebox(0,0)[t]{\lineheight{1.25}\smash{\begin{tabular}[t]{c}26\end{tabular}}}}%
    \put(0.53570815,0.01160838){\color[rgb]{0.14901961,0.14901961,0.14901961}\makebox(0,0)[t]{\lineheight{1.25}\smash{\begin{tabular}[t]{c}32\end{tabular}}}}%
    \put(0.56670522,0.01160838){\color[rgb]{0.14901961,0.14901961,0.14901961}\makebox(0,0)[t]{\lineheight{1.25}\smash{\begin{tabular}[t]{c}38\end{tabular}}}}%
    \put(0.59770229,0.01160838){\color[rgb]{0.14901961,0.14901961,0.14901961}\makebox(0,0)[t]{\lineheight{1.25}\smash{\begin{tabular}[t]{c}44\end{tabular}}}}%
    \put(0.50316136,-.01){\color[rgb]{0.14901961,0.14901961,0.14901961}\makebox(0,0)[t]{\lineheight{1.25}\smash{\begin{tabular}[t]{c}$100\rho$\end{tabular}}}}%
    \put(0,0){\includegraphics[width=\unitlength,page=6]{clusteringstudy1.pdf}}%
    \put(0.37518938,0.32306443){\color[rgb]{0.14901961,0.14901961,0.14901961}\makebox(0,0)[rt]{\lineheight{1.25}\smash{\begin{tabular}[t]{r}1\end{tabular}}}}%
    \put(0.37518938,0.27227405){\color[rgb]{0.14901961,0.14901961,0.14901961}\makebox(0,0)[rt]{\lineheight{1.25}\smash{\begin{tabular}[t]{r}2.7\end{tabular}}}}%
    \put(0.37518938,0.22148367){\color[rgb]{0.14901961,0.14901961,0.14901961}\makebox(0,0)[rt]{\lineheight{1.25}\smash{\begin{tabular}[t]{r}4.3\end{tabular}}}}%
    \put(0.37518938,0.1706933){\color[rgb]{0.14901961,0.14901961,0.14901961}\makebox(0,0)[rt]{\lineheight{1.25}\smash{\begin{tabular}[t]{r}6.0\end{tabular}}}}%
    \put(0.37518938,0.11990291){\color[rgb]{0.14901961,0.14901961,0.14901961}\makebox(0,0)[rt]{\lineheight{1.25}\smash{\begin{tabular}[t]{r}7.7\end{tabular}}}}%
    \put(0.37518938,0.06911253){\color[rgb]{0.14901961,0.14901961,0.14901961}\makebox(0,0)[rt]{\lineheight{1.25}\smash{\begin{tabular}[t]{r}9.4\end{tabular}}}}%
    \put(0.33866947,0.17973115){\color[rgb]{0.14901961,0.14901961,0.14901961}\rotatebox{90}{\makebox(0,0)[t]{\lineheight{1.25}\smash{\begin{tabular}[t]{c}$100\Delta$\end{tabular}}}}}%
    \put(0,0){\includegraphics[width=\unitlength,page=7]{clusteringstudy1.pdf}}%
    \put(0.64661917,0.02467095){\color[rgb]{0.14901961,0.14901961,0.14901961}\makebox(0,0)[lt]{\lineheight{1.25}\smash{\begin{tabular}[t]{l}0\end{tabular}}}}%
    \put(0.64661917,0.08950972){\color[rgb]{0.14901961,0.14901961,0.14901961}\makebox(0,0)[lt]{\lineheight{1.25}\smash{\begin{tabular}[t]{l}.1\end{tabular}}}}%
    \put(0.64661917,0.15434851){\color[rgb]{0.14901961,0.14901961,0.14901961}\makebox(0,0)[lt]{\lineheight{1.25}\smash{\begin{tabular}[t]{l}.2\end{tabular}}}}%
    \put(0.64661917,0.21918728){\color[rgb]{0.14901961,0.14901961,0.14901961}\makebox(0,0)[lt]{\lineheight{1.25}\smash{\begin{tabular}[t]{l}.3\end{tabular}}}}%
    \put(0.64661917,0.28402605){\color[rgb]{0.14901961,0.14901961,0.14901961}\makebox(0,0)[lt]{\lineheight{1.25}\smash{\begin{tabular}[t]{l}.4\end{tabular}}}}%
    \put(0,0){\includegraphics[width=\unitlength,page=8]{clusteringstudy1.pdf}}%
    \put(0.83964829,0.33654628){\makebox(0,0)[t]{\lineheight{1.25}\smash{\begin{tabular}[t]{c}p value ratio\end{tabular}}}}%
    \put(0,0){\includegraphics[width=\unitlength,page=9]{clusteringstudy1.pdf}}%
    \put(0.7168403,0.01160838){\color[rgb]{0.14901961,0.14901961,0.14901961}\makebox(0,0)[t]{\lineheight{1.25}\smash{\begin{tabular}[t]{c}1\end{tabular}}}}%
    \put(0.74793073,0.01160838){\color[rgb]{0.14901961,0.14901961,0.14901961}\makebox(0,0)[t]{\lineheight{1.25}\smash{\begin{tabular}[t]{c}7\end{tabular}}}}%
    \put(0.77902116,0.01160838){\color[rgb]{0.14901961,0.14901961,0.14901961}\makebox(0,0)[t]{\lineheight{1.25}\smash{\begin{tabular}[t]{c}13\end{tabular}}}}%
    \put(0.81011159,0.01160838){\color[rgb]{0.14901961,0.14901961,0.14901961}\makebox(0,0)[t]{\lineheight{1.25}\smash{\begin{tabular}[t]{c}20\end{tabular}}}}%
    \put(0.84120203,0.01160838){\color[rgb]{0.14901961,0.14901961,0.14901961}\makebox(0,0)[t]{\lineheight{1.25}\smash{\begin{tabular}[t]{c}26\end{tabular}}}}%
    \put(0.87229246,0.01160838){\color[rgb]{0.14901961,0.14901961,0.14901961}\makebox(0,0)[t]{\lineheight{1.25}\smash{\begin{tabular}[t]{c}32\end{tabular}}}}%
    \put(0.90338289,0.01160838){\color[rgb]{0.14901961,0.14901961,0.14901961}\makebox(0,0)[t]{\lineheight{1.25}\smash{\begin{tabular}[t]{c}38\end{tabular}}}}%
    \put(0.93447332,0.01160838){\color[rgb]{0.14901961,0.14901961,0.14901961}\makebox(0,0)[t]{\lineheight{1.25}\smash{\begin{tabular}[t]{c}44\end{tabular}}}}%
    \put(0.83964756,-.01){\color[rgb]{0.14901961,0.14901961,0.14901961}\makebox(0,0)[t]{\lineheight{1.25}\smash{\begin{tabular}[t]{c}$100\rho$\end{tabular}}}}%
    \put(0,0){\includegraphics[width=\unitlength,page=10]{clusteringstudy1.pdf}}%
    \put(0.71130215,0.32306443){\color[rgb]{0.14901961,0.14901961,0.14901961}\makebox(0,0)[rt]{\lineheight{1.25}\smash{\begin{tabular}[t]{r}1.0\end{tabular}}}}%
    \put(0.71130215,0.27227405){\color[rgb]{0.14901961,0.14901961,0.14901961}\makebox(0,0)[rt]{\lineheight{1.25}\smash{\begin{tabular}[t]{r}2.7\end{tabular}}}}%
    \put(0.71130215,0.22148367){\color[rgb]{0.14901961,0.14901961,0.14901961}\makebox(0,0)[rt]{\lineheight{1.25}\smash{\begin{tabular}[t]{r}4.3\end{tabular}}}}%
    \put(0.71130215,0.1706933){\color[rgb]{0.14901961,0.14901961,0.14901961}\makebox(0,0)[rt]{\lineheight{1.25}\smash{\begin{tabular}[t]{r}6.0\end{tabular}}}}%
    \put(0.71130215,0.11990291){\color[rgb]{0.14901961,0.14901961,0.14901961}\makebox(0,0)[rt]{\lineheight{1.25}\smash{\begin{tabular}[t]{r}7.7\end{tabular}}}}%
    \put(0.71130215,0.06911253){\color[rgb]{0.14901961,0.14901961,0.14901961}\makebox(0,0)[rt]{\lineheight{1.25}\smash{\begin{tabular}[t]{r}9.4\end{tabular}}}}%
    \put(0.67978227,0.17973115){\color[rgb]{0.14901961,0.14901961,0.14901961}\rotatebox{90}{\makebox(0,0)[t]{\lineheight{1.25}\smash{\begin{tabular}[t]{c}$100\Delta$\end{tabular}}}}}%
    \put(0,0){\includegraphics[width=\unitlength,page=11]{clusteringstudy1.pdf}}%
    \put(0.98347883,0.02467095){\color[rgb]{0.14901961,0.14901961,0.14901961}\makebox(0,0)[lt]{\lineheight{1.25}\smash{\begin{tabular}[t]{l}1\end{tabular}}}}%
    \put(0.98347883,0.1099616){\color[rgb]{0.14901961,0.14901961,0.14901961}\makebox(0,0)[lt]{\lineheight{1.25}\smash{\begin{tabular}[t]{l}10\end{tabular}}}}%
    \put(0.98347883,0.19525217){\color[rgb]{0.14901961,0.14901961,0.14901961}\makebox(0,0)[lt]{\lineheight{1.25}\smash{\begin{tabular}[t]{l}100\end{tabular}}}}%
    \put(0,0){\includegraphics[width=\unitlength,page=12]{clusteringstudy1.pdf}}%
  \end{picture}%
\endgroup%

%% file: nyc_adjacencyinput.tex
\begingroup%
  \makeatletter%
  \providecommand\color[2][]{%
    \errmessage{(Inkscape) Color is used for the text in Inkscape, but the package 'color.sty' is not loaded}%
    \renewcommand\color[2][]{}%
  }%
  \providecommand\transparent[1]{%
    \errmessage{(Inkscape) Transparency is used (non-zero) for the text in Inkscape, but the package 'transparent.sty' is not loaded}%
    \renewcommand\transparent[1]{}%
  }%
  \providecommand\rotatebox[2]{#2}%
  \newcommand*\fsize{\dimexpr\f@size pt\relax}%
  \newcommand*\lineheight[1]{\fontsize{\fsize}{#1\fsize}\selectfont}%
  \ifx\svgwidth\undefined%
    \setlength{\unitlength}{321.48509216bp}%
    \ifx\svgscale\undefined%
      \relax%
    \else%
      \setlength{\unitlength}{\unitlength * \real{\svgscale}}%
    \fi%
  \else%
    \setlength{\unitlength}{\svgwidth}%
  \fi%
  \global\let\svgwidth\undefined%
  \global\let\svgscale\undefined%
  \makeatother%
  \begin{picture}(1,0.80976524)%
    \lineheight{1}%
    \setlength\tabcolsep{0pt}%
    \put(0,0){\includegraphics[width=\unitlength,page=1]{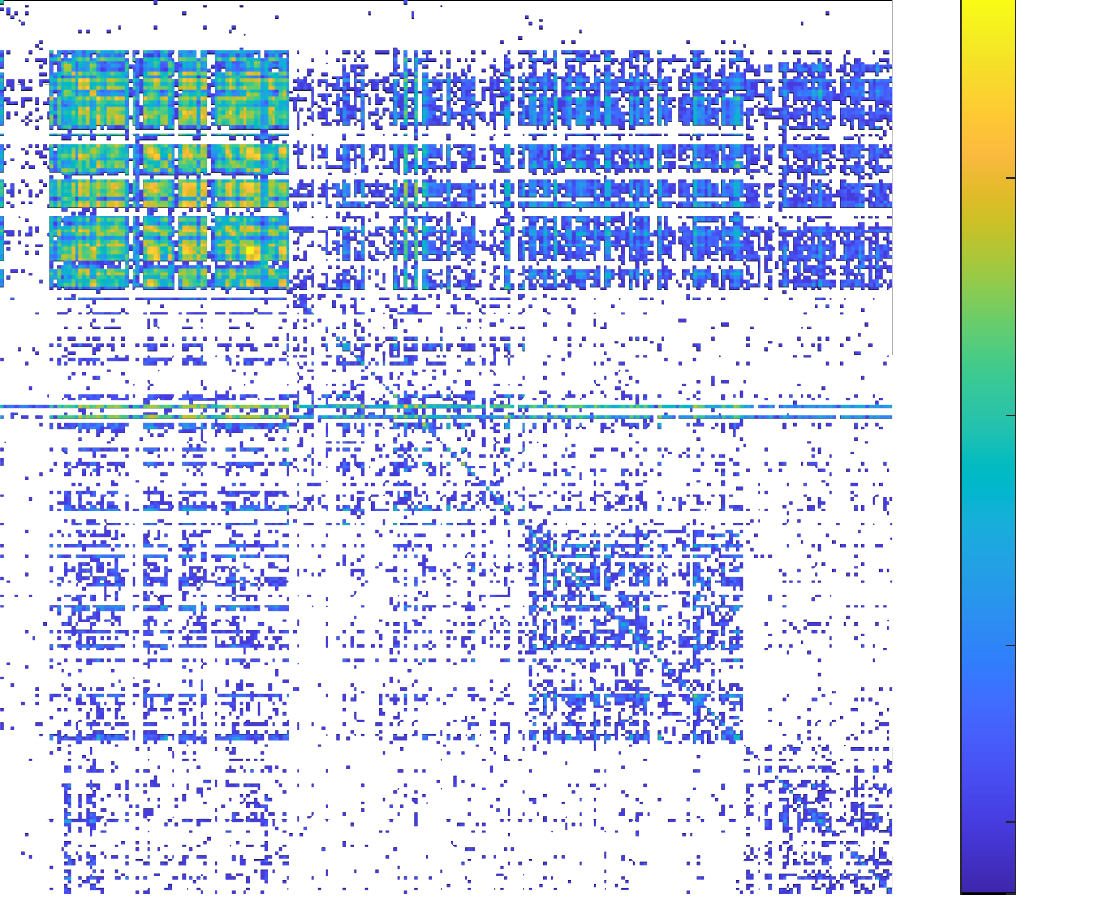}}%
    \put(0.92166023,0.00039641){\makebox(0,0)[lt]{\lineheight{1.25}\smash{\begin{tabular}[t]{l}0\end{tabular}}}}%
    \put(0.92166023,0.06479022){\makebox(0,0)[lt]{\lineheight{1.25}\smash{\begin{tabular}[t]{l}1\end{tabular}}}}%
    \put(0.92166023,0.22316281){\makebox(0,0)[lt]{\lineheight{1.25}\smash{\begin{tabular}[t]{l}10\end{tabular}}}}%
    \put(0.92166023,0.42914476){\makebox(0,0)[lt]{\lineheight{1.25}\smash{\begin{tabular}[t]{l}100\end{tabular}}}}%
    \put(0.92166023,0.64222509){\makebox(0,0)[lt]{\lineheight{1.25}\smash{\begin{tabular}[t]{l}1,000\end{tabular}}}}%
    \put(0,0){\includegraphics[width=\unitlength,page=2]{nyc_adjacency.pdf}}%
  \end{picture}%
\endgroup%

%% file: nyc_distanceinput.tex
\begingroup%
  \makeatletter%
  \providecommand\color[2][]{%
    \errmessage{(Inkscape) Color is used for the text in Inkscape, but the package 'color.sty' is not loaded}%
    \renewcommand\color[2][]{}%
  }%
  \providecommand\transparent[1]{%
    \errmessage{(Inkscape) Transparency is used (non-zero) for the text in Inkscape, but the package 'transparent.sty' is not loaded}%
    \renewcommand\transparent[1]{}%
  }%
  \providecommand\rotatebox[2]{#2}%
  \newcommand*\fsize{\dimexpr\f@size pt\relax}%
  \newcommand*\lineheight[1]{\fontsize{\fsize}{#1\fsize}\selectfont}%
  \ifx\svgwidth\undefined%
    \setlength{\unitlength}{301.52073669bp}%
    \ifx\svgscale\undefined%
      \relax%
    \else%
      \setlength{\unitlength}{\unitlength * \real{\svgscale}}%
    \fi%
  \else%
    \setlength{\unitlength}{\svgwidth}%
  \fi%
  \global\let\svgwidth\undefined%
  \global\let\svgscale\undefined%
  \makeatother%
  \begin{picture}(1,0.85483346)%
    \lineheight{1}%
    \setlength\tabcolsep{0pt}%
    \put(0,0){\includegraphics[width=\unitlength,page=1]{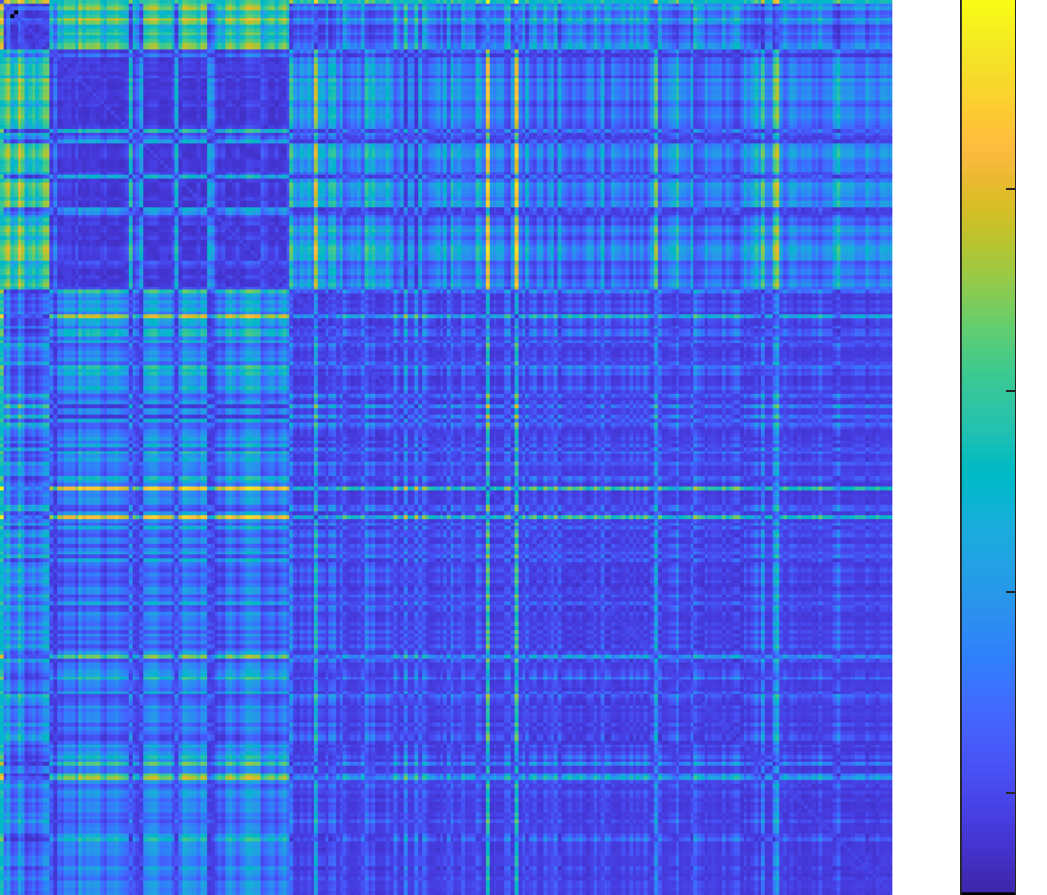}}%
    \put(0.9826854,0.08821069){\makebox(0,0)[lt]{\lineheight{1.25}\smash{\begin{tabular}[t]{l}1\end{tabular}}}}%
    \put(0.9826854,0.28048503){\makebox(0,0)[lt]{\lineheight{1.25}\smash{\begin{tabular}[t]{l}2\end{tabular}}}}%
    \put(0.9826854,0.47275937){\makebox(0,0)[lt]{\lineheight{1.25}\smash{\begin{tabular}[t]{l}3\end{tabular}}}}%
    \put(0.9826854,0.66503395){\makebox(0,0)[lt]{\lineheight{1.25}\smash{\begin{tabular}[t]{l}4\end{tabular}}}}%
    \put(0,0){\includegraphics[width=\unitlength,page=2]{nyc_distance.pdf}}%
  \end{picture}%
\endgroup%

%% file: nyc_d1input.tex
\begingroup%
  \makeatletter%
  \providecommand\color[2][]{%
    \errmessage{(Inkscape) Color is used for the text in Inkscape, but the package 'color.sty' is not loaded}%
    \renewcommand\color[2][]{}%
  }%
  \providecommand\transparent[1]{%
    \errmessage{(Inkscape) Transparency is used (non-zero) for the text in Inkscape, but the package 'transparent.sty' is not loaded}%
    \renewcommand\transparent[1]{}%
  }%
  \providecommand\rotatebox[2]{#2}%
  \newcommand*\fsize{\dimexpr\f@size pt\relax}%
  \newcommand*\lineheight[1]{\fontsize{\fsize}{#1\fsize}\selectfont}%
  \ifx\svgwidth\undefined%
    \setlength{\unitlength}{307.15892029bp}%
    \ifx\svgscale\undefined%
      \relax%
    \else%
      \setlength{\unitlength}{\unitlength * \real{\svgscale}}%
    \fi%
  \else%
    \setlength{\unitlength}{\svgwidth}%
  \fi%
  \global\let\svgwidth\undefined%
  \global\let\svgscale\undefined%
  \makeatother%
  \begin{picture}(1,0.83914221)%
    \lineheight{1}%
    \setlength\tabcolsep{0pt}%
    \put(0,0){\includegraphics[width=\unitlength,page=1]{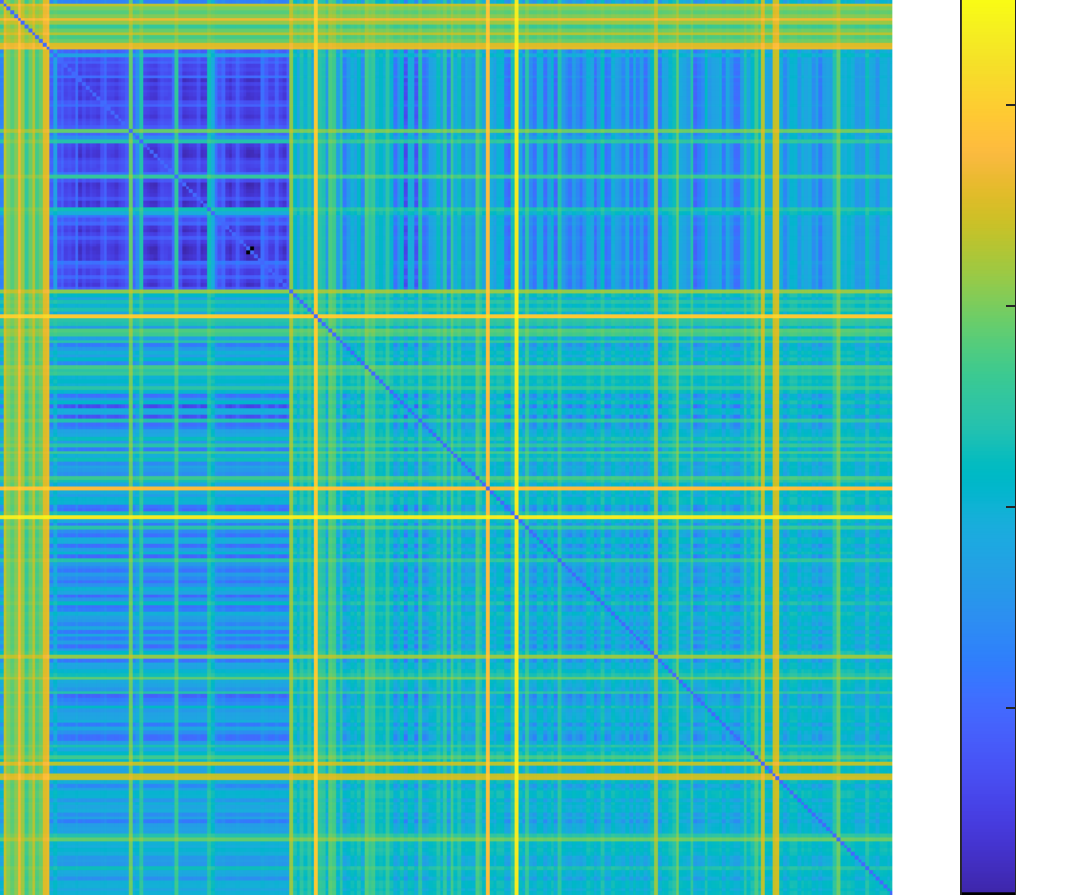}}%
    \put(0.96464731,0.16627135){\makebox(0,0)[lt]{\lineheight{1.25}\smash{\begin{tabular}[t]{l}6\end{tabular}}}}%
    \put(0.96464731,0.35471256){\makebox(0,0)[lt]{\lineheight{1.25}\smash{\begin{tabular}[t]{l}8\end{tabular}}}}%
    \put(0.96464731,0.54315353){\makebox(0,0)[lt]{\lineheight{1.25}\smash{\begin{tabular}[t]{l}10\end{tabular}}}}%
    \put(0.96464731,0.73159475){\makebox(0,0)[lt]{\lineheight{1.25}\smash{\begin{tabular}[t]{l}12\end{tabular}}}}%
    \put(0,0){\includegraphics[width=\unitlength,page=2]{nyc_d1.pdf}}%
  \end{picture}%
\endgroup%